\documentclass[12pt]{article}
\usepackage{sectsty}
\usepackage[T1]{fontenc}
\usepackage{kpfonts}
\usepackage[dvipsnames]{xcolor}
\usepackage{pdfpages}
\usepackage{sgame}
\usepackage{accents}
\usepackage{tikz-cd}
\usepackage{float}
\usepackage[round]{natbib}
\usepackage{multirow}
\usepackage{dutchcal}
\usepackage{pdfpages}
\usepackage{bbm}
\usepackage{sgame}
\usepackage{mathtools}
\usepackage{amsthm,thm-restate}
\usepackage{enumitem}
\usepackage[margin=1.25in]{geometry}
\usepackage{caption}
\usepackage{subcaption}
\usepackage{epigraph}
\usepackage{listings}
\usetikzlibrary{calc}
\usetikzlibrary{shapes,arrows}
\usepackage{tcolorbox}

\newcommand{\ubar}[1]{\underaccent{\bar}{#1}}
\DeclareMathOperator\supp{supp}

\DeclareMathOperator\conv{conv}

\DeclareMathOperator*{\argmax}{arg\,max}

\newtheorem{theorem}{Theorem}[section]
\newtheorem{proposition}[theorem]{Proposition}
\newtheorem{lemma}[theorem]{Lemma}
\newtheorem{corollary}[theorem]{Corollary}
\newtheorem{claim}[theorem]{Claim}

\theoremstyle{definition}

\newtheorem{definition}[theorem]{Definition}
\newtheorem{example}[theorem]{Example}

\newtheorem{remark}[theorem]{Remark}

\sectionfont{\color{BlueViolet}}
\subsectionfont{\color{BlueViolet}}

\usepackage{hyperref}
\hypersetup{
    colorlinks=true,
    linkcolor=OrangeRed,
    filecolor=Periwinkle,      
    urlcolor=Periwinkle,
    citecolor=Mulberry,
}

\definecolor{backcolour}{rgb}{0.63, 0.79, 0.95}

\lstdefinestyle{mystyle}{
  backgroundcolor=\color{backcolour},
  basicstyle=\ttfamily\footnotesize,
  breakatwhitespace=false,         
  breaklines=true,                 
  captionpos=b,                    
  keepspaces=true,                 
  numbers=left,                    
  numbersep=5pt,                  
  showspaces=false,                
  showstringspaces=false,
  showtabs=false,                  
  tabsize=2
}


\begin{document}
\title{Blackwell-Monotone Updating Rules}
\author{Mark Whitmeyer\thanks{Arizona State University. Email: \href{mailto:mark.whitmeyer@gmail.com}{mark.whitmeyer@gmail.com}. Dedicated to KS. I thank the editor, Emir Kamenica, and four anonymous referees for their comments. Aislinn Bohren, Faruk Gul, Vasudha Jain, Ehud Lehrer, Benson Tsz Kin Leung, Paula Onuchic, Eddie Schlee, Ludvig Sinander, Shoshana Vasserman, Joseph Whitmeyer, Tom Wiseman, Kun Zhang and seminar and conference audiences at CETC 2023 (Simon Fraser), HKBU/Kyoto/NTU/Osaka/Sinica, UC Berkeley, UC Davis, and the University of Pennsylvania provided invaluable feedback. I am also grateful to Tong Wang, whose question inspired this paper. Nour Chalhoub and Konstantin von Beringe provided excellent research assistance. This paper was previously titled ``Bayes \(=\) Blackwell, Almost.''}}

\date{\today}

\maketitle

\begin{abstract}An updating rule specifies how an agent reacts to information. An updating rule is Blackwell monotone if more information is always better for an agent in a decision problem and strictly Blackwell monotone if, in addition, there is always a decision problem in which more information is strictly better for an agent. Bayes' law is strictly Blackwell monotone, and I show that within a broad class of updating rules--those that distort the Bayesian posteriors in a signal-independent manner--it is the only strictly Blackwell-monotone updating rule. If an agent's decisions are evaluated non-paternalistically (according to her beliefs), the Blackwell-monotone updating rules are affine distortions of the Bayesian posteriors. \end{abstract}

\newpage

\section{Introduction}

The Blackwell order over statistical experiments (\cite{blackwell} and \cite{blackwell2}) provides an elegant way of comparing information structures. To a Bayesian expected-utility maximizing agent, if one experiment ranks above another in the Blackwell order, it yields a higher welfare to the agent than the other, no matter the agent's decision problem. We term this property possessed by Bayesian updating--that Blackwell-more-informative experiments yield higher payoffs to an agent no matter her decision problem--\textit{Blackwell monotonicity}. If an agent violates Blackwell monotonicity, she cannot be a Bayesian expected-utility maximizer.

This paper asks the converse question: is it true that if an agent is not Bayesian, she must violate Blackwell monotonicity? There is an important subtlety to this query. For a Bayesian updater, there is no ambiguity about how the value of information and, hence, Blackwell monotonicity is defined. In contrast, it is not clear how the value of information should be defined for a non-Bayesian. At the interim point, when the agent is taking an action based on a non-Bayesian belief, should we compute her expected utility with the Bayesian posterior, or the one the agent herself is using? It may well be that the agent also makes forecasting errors about the likelihood of various pieces of information. From this \textit{ex ante} view, whose perspective should we take? 

Our two main results, Theorems \ref{maintheorem} and \ref{secondmaintheorem}, characterize Blackwell monotonicity when we evaluate an agent's interim behavior either according to the Bayesian belief (paternalistically) or according to the agent's belief (non-paternalistically). To distinguish between the two concepts, we term the latter \textit{non-paternalistic Blackwell monotonicity}.\footnote{We characterize Blackwell monotonicity when welfare is computed with an alternative forecasting rule in Appendix \ref{forecast}.}

When an agent's decisions are evaluated paternalistically, updating in a non-Bayesian manner may lead her to make decisions that are strictly suboptimal from the Bayesian's point of view. Mistaken beliefs may lead to mistaken actions, and these errors can be so costly that the agent is worse off from some information than if she had obtained no information. For example, one simple updating mistake is overreaction to information, in which an agent's beliefs depart too far from the prior than they should under Bayes' law. Even in simple decision problems with just two actions, overreaction can lead to some information being strictly worse than none, as the agent will be too eager to abandon the prior-optimal action. Thus, the agent's updating is not Blackwell monotone.

Another variety of mistake an agent might make is to underreact to information. Now, her posteriors are closer to the prior than they should be under Bayes' law. Surprisingly, for some varieties of underreaction, it is not the case that some information must be worse than none, even \textit{if we evaluate decisions according to the correct Bayesian posterior}. As \cite{morris1997rationality} and \cite{braghieri2023biased} reveal, a broad variety of updating rules are such that some information can never be worse than no information, even when we evaluate an agent's welfare paternalistically. Nevertheless, as we will see in \S\ref{ynot}'s example, an underreacting agent will nevertheless violate Blackwell monotonicity, as she may fail to take advantage of additional information.

When we evaluate an agent's choice of action non-paternalistically, it is clear that some information is always better than no information. At any posterior, if an agent picks an action other than one that is optimal at the prior, the new action must yield her a higher expected payoff at that belief than the prior-optimal ones. But, if we evaluate actions non-paternalistically, we agree with her! Consequently, some information is always better than none. On the other hand, as in the paternalistic case, non-paternalistic Blackwell monotonicity requires that the agent take advantage of additional information. As we will see, this is satisfied only by a particular variety of updating rules.

For the first half of this paper, we restrict attention to a special class of updating rules--those that systematically distort beliefs, i.e., for which an agent's posterior on observing an experiment realization is a distortion of the Bayesian posterior that is otherwise independent of the experiment itself. We also refine Blackwell monotonicity (both paternalistic and non-paternalistic) by introducing a ``strict'' modifier: an updating rule is strictly Blackwell monotone if it is not only Blackwell monotone but such that more information is strictly more valuable in some decision problem.

In the first of the two main results of this paper, Theorem \ref{maintheorem}, we show that in the class of updating rules that systematically distort beliefs, Bayes' law is the unique strictly Blackwell-monotone updating rule. That is, if an agent always prefers more information, and there is always a decision problem in which strictly more information can be exploited, she must be a Bayesian. In the second main result, Theorem \ref{secondmaintheorem}, we show that in the class of updating rules that systematically distort beliefs, the strictly non-paternalistic Blackwell-monotone updating rules are precisely those in which the distortion is a constant scaling of the Bayesian posterior plus a translation, \textit{viz.}, the distortion is affine. Although non-paternalistic Blackwell monotonicity demands this extreme sense of consistency in how the agent distorts Bayesian beliefs, we see that the value of information may still be positive for a non-Bayesian agent, \textit{but only from her perspective.}

We then push the limits of Theorem \ref{maintheorem} by moving beyond systematic distortions to find that there are two broader classes of updating rules in which Bayesianism is equivalent to strict Blackwell monotonicity. An updating rule is \textit{convex} if pooling experiments cannot increase their \textit{ex ante} value.\footnote{That is, an updating rule is convex if an agent always (weakly) prefers to observe experiment \(\pi\) with probability \(\lambda\) and \(\pi'\) with its complement--that is, randomize over two experiments--than observe their convex combination, experiment \(\lambda \pi + \left(1-\lambda\right) \pi'\), for sure.} Under Bayes’ law, a pooled experiment is worth exactly the probability‑weighted average of its components, so Bayes' law is convex--indeed, linear. An updating rule is \textit{grounded} if the absence of information (a completely uninformative experiment) is updated correctly. Proposition \ref{convexprop} reveals that a grounded, convex updating rule is strictly Blackwell monotone if and only if it is Bayes' law. \textit{Focused} updating rules generalize those that systematically distort beliefs: such rules are those for which identically-generated experiment realizations are interpreted identically. In Proposition \ref{focusprop}, we find that a grounded, focused updating rule is strictly Blackwell monotone if and only if it is Bayes' law. Each of these environments makes underreaction incompatible with Blackwell monotonicity, as the agent may fail to exploit additional information.

These various additional structures matter, in that groundedness alone does not pin down Bayes' law under strict Blackwell monotonicity. 
In \S\ref{nonbayesex}, we construct a strictly Blackwell-monotone and grounded rule other than Bayes' law. Importantly, however, by our equivalence results, it does not satisfy the other desiderata of convexity or focus. Thus, one perspective on these findings is that they provide an impossibility result: strict Blackwell-monotonicity and one of these other criteria cannot jointly be satisfied by any updating rule other than Bayes' law. Even within the class of systematic distortions (and even if we require that the distortion be continuous), the ``strict'' modifier is also important. In particular, when there are just two states, an updating rule that systematically distorts posteriors according to a continuous distortion is Blackwell monotone if and only if it corresponds to extreme-belief aversion (Appendix \ref{appx}).

These results are also useful in that they remove the need to check whether a particular specification of a (systematically distorting) non-Bayesian updating rule is Blackwell monotone. Unless there are two states and it is some form of extreme-belief aversion, any non-Bayesian updating rule that systematically distorts posteriors is such that there exist two Blackwell-ranked experiments with the property that the less-informative of the two is strictly better for the agent (evaluated paternalistically). Thus, case-by-case scrutiny of specific forms of non-Bayesian updating, like the conservative Bayesians of \cite{edwards}, confirmatory bias (\cite{rabinschrag}), divisible updating (\cite{divisibleupdating}), and \cite{grether}'s \(\alpha-\beta\) model, is unnecessary. None are Blackwell-monotone. It is also easy to check whether a rule is affine, making non-paternalistic Blackwell monotonicity easy to discern. Finally, these findings also motivate the construction of orders over experiments for non-Bayesians. How might we (paternalistically) rank experiments for an underreacter, for instance? How might underreacter rank experiments herself?

\subsection{Related Work}
By now, many papers explore non-Bayesian updating. One vein of the literature formulates axioms that produce updating rules other than Bayes' law. \cite{epstein06} studies an agent whose behavior is in the spirit of the prone-to-temptation agent of \cite{gulpes}. \cite{ortoleva} axiomatizes a model in which an agent does not behave as a perfect Bayesian when confronted with unexpected news.

Of special note is \cite{jakob}, who introduces a model of coarse Bayesian updating in which a decision-maker (DM) partitions the belief simplex into a collection of convex sets. \cite{jakob} presents an example (Example 4) of a coarse Bayesian updater who, nevertheless, assigns a higher value to more information. This is a particular case of extreme-belief aversion, the unique family of rules that are Blackwell monotone when there are two states (Appendix \ref{appx}). His Proposition 7 states precisely when a regular (for which all cells of the partition have full dimension) coarse Bayesian updating rule is Blackwell monotone.

There has also been a recent flurry of papers investigating various questions in this realm. \citet{molavi2021empirical} characterizes the empirical content of Bayesian updating. \cite{ba2024} explore how cognitive constraints shape belief updating, yielding underreaction in simpler environments and overreaction in complex ones. \cite{bohren2023} dichotomize biased updating into retrospective and prospective components; the former of which is central to this paper's analysis. \cite{chambers2024} propose a consistency criterion for belief distortions that singles out exponential distortions \`{a} la \cite{grether}. \cite{frick2024} study the long-run consequences of learning biases and construct both static and dynamic welfare rankings of updating rules. \cite{bordoli2024} also examines how non-Bayesian updating affects the value of information and describes the essential features of two natural versions of this concept.

This work is also related to the work on dynamically consistent beliefs--see, e.g., \cite{gul1990,machina92,border94,siniscalchi2011}; and the survey, \cite{machina89}. Another seminal paper in that area is \cite{epstein}, who show that ``dynamically consistent beliefs must be Bayesian,'' thereby establishing an equivalence (as Bayesian beliefs are obviously dynamically consistent). 
We find that within our specified classes, dynamic consistency and strict Blackwell monotonicity are also equivalent.

Closely tied to the notion of dynamic consistency is the value of information for DMs with non-expected-utility preferences. That some experiments may be harmful to a DM is illustrated in \cite{wakker}, \cite{hilton1990}, \cite{safra1995}, and \cite{hill2020}. \cite{li2016} show that the Blackwell order holds for almost all DMs with uncertainty-averse preferences provided they can commit \textit{ex ante} to actions, and \cite{celen} establishes that it holds for an MEU DM.

\subsection{Why Not Underreact?}\label{ynot}

When an agent's welfare is evaluated paternalistically, not all departures from Bayesian updating are equally harmful. There is a strong sense in which \textit{underreaction}, where a DM's posterior is a convex combination of the Bayesian posterior and the prior, is superior to \textit{overreaction}, where the Bayesian posterior is a convex combination of the DM's posterior and the prior. As \cite{morris1997rationality} and \cite{braghieri2023biased} show, a DM who underreacts can never be hurt by information: some information is always better than none.\footnote{\cite{morris1997rationality} and \cite{braghieri2023biased} provide a full characterization of updating rules that lead to a positive value for information (versus none). \cite{von2024perils} have a simpler and closely related thesis--underreacters can never be exploited by a malevolent principal whereas overreacters always can.} In stark contrast, for an agent who overreacts, there always exists a decision problem in which her losses from information (versus none) are unboundedly large.

It is easy to see why overreaction can be so harmful: the DM can be tricked into taking an extreme action that she should not take, and if incentives are steep enough, this mistake will be very costly. The innocuity of underreaction is more subtle: essentially, the martingality of the Bayesian posteriors means that a decision problem in which a DM makes mistakes at some beliefs must be one in which she benefits at others, and this benefit must outweigh the cost of any mistakes. But Blackwell-monotonicity requires something stronger: more information must always be preferred to less, not merely none.

We now illustrate these ideas in the context of a simple two-state example. The state space, \(\Theta\), is binary, \(\Theta = \left\{0,1\right\}\) and the DM's prior is \(\mu_0 = \mathbb{P}\left(1\right) = 1/2\). The DM has a binary decision problem with action set \(A = \left\{a_0,a_1\right\}\). Action \(a_0\) yields a state-independent payoff of \(0\) (the DM's utility is \(u(a_0,\theta) = 0\) for all \(\theta \in \Theta\)), whereas action \(a_1\) yields payoff \(-30\) in state \(0\) and \(20\) in state \(1\) (\(u(a_1,0) = -30\) and \(u(a_1,1) = 20\)). The DM is indifferent between the two actions when her belief \(\mu \equiv \mathbb{P}(1)\) equals \(3/5\).

Now take a binary experiment with realizations \(L\) and \(H\), where the Bayesian posterior upon observing \(L\) is \(0\) and the Bayesian posterior upon observing \(H\) is some \(\mu_H > 1/2\). If the DM overreacts, this signal can be worse for her than no information. In particular, suppose \(\mu_H < 3/5\), in which case a Bayesian DM will never take action \(a_1\), which yields her a payoff of \(0\), the same as her payoff to no information. However, if she overreacts and instead updates the high signal realization to some posterior that is strictly larger than \(3/5\), she will mistakenly take the high action, producing a strictly negative \textit{ex ante} expected payoff.

\begin{figure}
    \centering
    \includegraphics[width=0.6\linewidth]{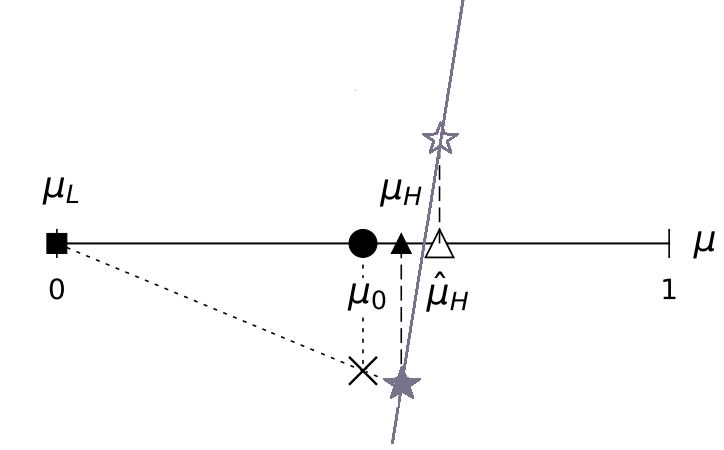}
    \caption{A Pernicious Overreaction}
    \label{fig1}
\end{figure}

This is depicted in Figure \ref{fig1}. The black line segment is the DM's payoff (in her belief \(\mu\)) to action \(a_0\). The gray line segment is her payoff to \(a_1\). The DM's prior (\(1/2\)) is the solid circle, and the low and high Bayesian posteriors, \(\mu_L = 0\) and \(\mu_H = 9/16\), are the solid square and triangle, respectively. Her expected payoff at belief \(\mu_H\) from taking action \(a_1\) is specified by the solid gray star and is strictly negative. On the other hand, the hollow triangle is her non-Bayesian posterior \(\hat{\mu}_H = 5/8\). Her expected payoff at this belief corresponds to the hollow gray star, which leads her to incorrectly take action \(a_1\) following realization \(H\), as the gray line segment lies above the black one at this belief.

With underreaction, the worst that can happen is a false negative: it may be that the Bayesian posterior following the high signal realization, \(\mu_H\), is strictly larger than \(3/5\), in which case the DM ought to take action \(a_1\). If, instead, the DM's posterior following the high realization is less than \(3/5\), she will persist with the \textit{status quo}, taking the prior-optimal action of \(a_0\). But this hurts her only in the sense that she is failing to take advantage of an opportunity--she is still no worse off than if she had had observed nothing.

This failure to take advantage of information is what makes underreaction violate Blackwell monotonicity. Suppose the DM underreacts in the sense that her updating rule produces posteriors that are \(50/50\) averages of the Bayesian posteriors with the prior. Take a signal that produces a distribution over Bayesian posteriors that is supported on \(\left\{0, 5/8, 7/8\right\}\) (with respective probabilities \(1/3\), \(1/3\), and \(1/3\)). The DM's updating produces posterior \(1/4\) instead of \(0\), \(9/16\) instead of \(5/8\), and \(11/16\) instead of \(7/8\). Consequently, the DM takes action \(a_1\) if and only if she observes the highest signal realization. She fails to fully take advantage of information because she should also take \(a_1\) following the intermediate signal realization but does not (as \(9/16 < 3/5 < 5/8\)).

The specified ternary distribution corresponds to a strictly more informative experiment than a binary one with support \(\left\{0,3/4\right\}\), yet yields a strictly lower \textit{ex ante} expected payoff to the DM. This is because the binary experiment does not produce any consequential mistakes: although the high realization in the binary is underreacted to, the DM's posterior (\(5/8\)) is still high enough for her to pick the correct action at the interim point.

\begin{figure}
    \centering
    \includegraphics[width=0.7\linewidth]{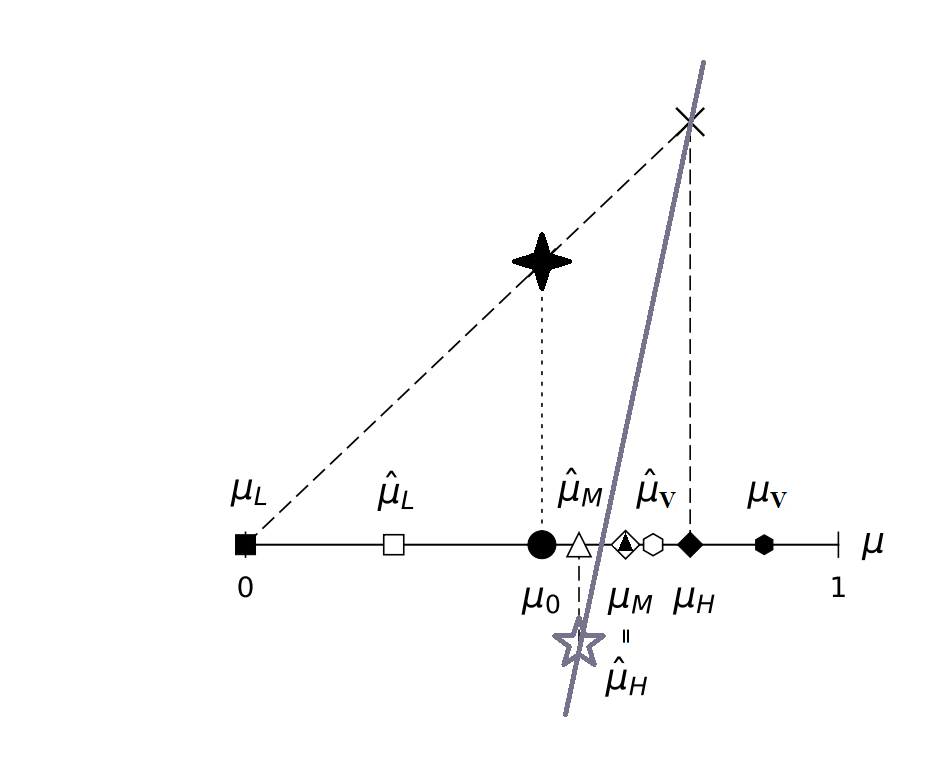}
    \caption{Underreaction Leads to Type-II Errors}
    \label{fig2}
\end{figure}

This is depicted in Figure \ref{fig2}. As before, the prior is the solid circle and the low posterior is the solid square \(\mu_L = 0\). The solid triangle is \(\mu_M = 5/8\), the solid diamond is \(\mu_H = 3/4\), and the solid hexagon is \(\mu_V = 7/8\). The underreacting \(DM\)'s posteriors are their hollow twins: \(\hat \mu_L\) is the hollow square, \(\hat{\mu}_M\) is the hollow triangle, \(\hat{\mu}_H\) (which equals \(\mu_M\)) is the hollow diamond, and \(\hat{\mu}_V\) is the hollow hexagon. The compass star is the Bayesian DM's \textit{ex ante} expected payoff from each experiment--the extra information is superfluous. The non-Bayesian DM, however, makes a mistake: she incorrectly believes that her expected payoff following realization \(M\) is the strictly-negative hollow gray star, which she wishes to avoid, so she will mistakenly choose \(a_0\).

\section{Setup} There is a finite set of states of nature, \(\Theta\), with \(\left|\Theta\right| = n \geq 2\). \(\Delta \equiv \Delta\left(\Theta\right)\) is the \(\left(n-1\right)\)-simplex, the set of probabilities on \(\Theta\), understood as a subset of \(\mathbb{R}^{n-1}\), topologized by the Euclidean metric. \(\mu_0 \in \Delta^{\circ}\) denotes our decision-maker's (DM's) full-support prior, where \(\Delta^{\circ}\) denotes the (topological) interior of \(\Delta\). For any belief \(\mu \in \Delta\), we write \(\mu\left(\theta\right) \equiv \mathbb{P}_\mu\left(\theta\right)\). 

A (statistical) experiment is a pair \(\left(\pi,S\right)\), consisting of a stochastic map \(\pi \colon \Theta \to \Delta \left(S\right)\) and a finite set of signal realizations \(S\) for which i. the unconditional probability of each \(s \in S\) is strictly positive; and ii. the Bayesian posterior after observing each \(s\), \(\mu_s\), is unique (for \(s,s' \in S\), if \(s \neq s'\), \(\mu_s \neq \mu_{s'}\)).\footnote{Neither assumption affects the results of this paper and both help simplify the analysis.} We note the formula for the Bayesian posterior: for all \(\theta \in \Theta\), \[\mu_s\left(\theta; \pi\right) = \frac{\mu_0\left(\theta\right) \pi\left(\left.s\right|\theta\right)}{\sum_{\theta' \in \Theta} \mu_0\left(\theta'\right) \pi\left(\left.s\right|\theta'\right)}\text{.}\]
\(\Pi\) denotes the set of experiments. In turn, from the \textit{ex ante} perspective a statistical experiment induces a distribution over Bayesian posteriors \(\rho \in \Delta \Delta\left(\Theta\right)\), where \(\mathbb{P}_{\rho}\left(\mu_s\right) = \rho(\mu_s) = \sum_{\theta' \in \Theta} \mu_0\left(\theta'\right) \pi\left(\left.s\right|\theta'\right)\).

We are interested in studying departures from Bayes' law. To that end, we denote the DM's posterior upon observing signal realization \(s\), given experiment \(\left(\pi,S\right)\), by \(\hat{\mu}_s\left(\theta; \pi\right)\).\footnote{Throughout, beliefs with hats will be those produced by the (potentially non-Bayesian) updating rule, and those without hats will be those produced by Bayes' law.} An \textcolor{OrangeRed}{Updating Rule}, \(U\), is a map
\[\begin{split}
   U \colon \Pi &\to \Delta^S\\
    \left(\pi,S\right) &\mapsto \left(\hat{\mu}_s\left(\theta; \pi\right)\right)_{s \in S}\text{,}
\end{split}\]
i.e., it maps a experiment to a collection of posteriors, one for each signal realization.

Following \cite{declippel}, we say that an updating rule \textcolor{OrangeRed}{Systematically Distorts Beliefs} if there exists a function \(\varphi \colon \Delta \to \Delta\) such that for all statistical experiments \(\left(\pi,S\right) \in \Pi\), and signal realizations \(s \in S\), 
\[\hat{\mu}_s\left(\theta; \pi\right) = \varphi\left(\mu_s\left(\theta; \pi\right)\right)\text{.}\] 
For most of the paper, we restrict attention to updating rules that systematically distort beliefs, which from now on we term \textcolor{OrangeRed}{Updating Rules}. This allows us to drop subscripts and arguments: given distortion \(\varphi\) and  a Bayesian posterior \(\mu\), the DM's posterior is \(\hat{\mu} \equiv \varphi(\mu)\).

A \textcolor{OrangeRed}{Decision Problem} \(\left(A, u\right)\) consists of a compact set of actions \(A\) and a continuous utility function \(u \colon A \times \Theta \to \mathbb{R}\). For a Bayesian DM, any decision problem induces a value function in the posterior \(\mu\) \[V\left(\mu\right) \equiv \max_{a \in A}
\mathbb{E}_{\mu}u\left(a,\theta\right) \text{ .}\]
\(V\) is convex, which implies a positive value of information for a Bayesian DM. In this paper, we are interested in evaluating the value of information for non-Bayesians. For simplicity, we specify that at every \(\hat{\mu} \in \Delta\) the DM's choice of action is \textcolor{OrangeRed}{Consistent}: her choice depends only on the realized posterior, i.e., selection \(a^*\left(\hat{\mu}\right) \in \argmax_{a \in A}\mathbb{E}_{\hat{\mu}} u\left(a,\theta\right)\) is a function of \(\hat{\mu}\).\footnote{Naturally, we are also assuming \textit{interim} optimality: \(a^*\) is a selection from the \(\argmax\) correspondence at belief \(\hat{\mu}\).}

For a fixed decision problem and consistent choice of action, the function \(W\left(\mu\right) \coloneqq \mathbb{E}_{\mu}u\left(a^*\left(\varphi\left(\mu\right)\right),\theta\right)\) is a well-defined function of the Bayesian posterior \(\mu\). This is a paternalistic way of evaluating an agent's welfare: we are evaluating her payoff not with the belief she uses to take her decision, \(\varphi(\mu)\), but with the Bayesian belief \(\mu\). In contrast, \(V(\varphi(\mu))\) is the agent's welfare when evaluated non-paternalistically. Letting \(\rho\) be the Bayesian distribution over posteriors corresponding to experiment \(\pi\), the \textit{ex ante} expected utility for the DM from observing \(\pi\) is \(\mathbb{E}_\rho W\left(\mu\right)\) if evaluated paternalistically, and \(\mathbb{E}_\rho V\left(\varphi\left(\mu\right)\right)\) if evaluated non-paternalistically. 

We can now state the central definitions of the paper, where \(\rho'\) is the Bayesian distribution over posteriors corresponding to \(\pi'\).
\begin{definition}
    An updating rule is \textcolor{OrangeRed}{Blackwell Monotone} if for any decision problem \(\left(A, u\right)\), consistent \(a^* \colon \Delta \to A\), and pair \(\pi \succeq \pi'\), \[\mathbb{E}_{\rho} W\left(\mu\right) \geq \mathbb{E}_{\rho'} W\left(\mu\right)\text{;}\]
    and \textcolor{OrangeRed}{Non-paternalistically Blackwell Monotone} if for any decision problem \(\left(A, u\right)\), consistent \(a^* \colon \Delta \to A\), and pair \(\pi \succeq \pi'\), \[\mathbb{E}_{\rho} V\left(\varphi\left(\mu\right)\right) \geq \mathbb{E}_{\rho'} V\left(\varphi\left(\mu\right)\right)\text{.}\]
\end{definition}
Blackwell monotonicity is equivalent to \(W\)'s convexity in \(\mu\), whereas non-paternalistic Blackwell monotonicity is equivalent to \(V \circ \varphi's\) convexity in \(\mu\).
\begin{definition}
    An updating rule is \textcolor{OrangeRed}{Strictly Blackwell Monotone} if it is  Blackwell monotone and for any pair \(\pi \succ \pi'\) there exists a decision problem and a consistent decision rule such that
\[\mathbb{E}_{\rho} W\left(\mu\right) > \mathbb{E}_{\rho'} W\left(\mu\right)\text{;}\]
and \textcolor{OrangeRed}{Strictly Non-paternalistically Blackwell Monotone} if it is non-paternalistically Blackwell monotone and for any pair \(\pi \succ \pi'\) there exists a decision problem and a consistent decision rule such that
\[\mathbb{E}_{\rho} V\left(\varphi\left(\mu\right)\right) >\mathbb{E}_{\rho'} V\left(\varphi\left(\mu\right)\right)\text{.}\]
\end{definition}
A third way of evaluating a non-Bayesian DM's value of information is to allow for the forecast distribution over signal realizations to also be incorrect--some Bayes-plausible \(\hat{\rho}\) instead of \(\rho\)--and so we evaluate \(\mathbb{E}_{\hat{\rho}} V\left(\hat{\mu}\right)\). We provide a full characterization of Blackwell monotonicity for this conception of it in Appendix \ref{forecast}.

\subsection{Why Systematic Distortions?}

The primary focus of this paper is on updating rules that systematically distort beliefs. There are several justifications for this. First, systematic distortions are those for which counterfactual events do not affect decision-making either directly or indirectly: only the state-dependent probabilities under the experiment of the signal realization matter for the DM's posterior, and the agent can simplify fractions, which renders the probabilities of the other signal realizations totally irrelevant \textit{ex interim}. This is consistent with our method of evaluating the value of information in that the DM encounters the signal realization and takes an interim-optimal action at her resulting posterior. Why should counterfactual events matter in her calculus?

Second, the bulk of the existing literature on updating rules focuses on those that systematically distort beliefs. Thus, our results concern a particularly relevant and useful class of rules. A third argument in favor of our emphasis on systematic distortions is that the subsequent results concerning focused (and grounded) rules and convex (and grounded) rules also single out Bayes' law as equivalent to strict Blackwell monotonicity. In this light, the systematic-distortion specification can be seen as a slightly simplified environment that is extremely tractable and yet whose properties extend more broadly.

\section{Blackwell-Monotone Updating Rules}

In this section, we state and discuss the two main results of the paper, which characterize the updating rules that are strictly Blackwell monotone and strictly non-paternalistically Blackwell monotone.  It is important to keep in mind that in this section, by updating rule, we mean “updating rule that systematically distorts beliefs,” though we also remind ourselves of this in the statements of the two main theorems.
\begin{restatable}{theorem}{maintheorem}
   \label{maintheorem}
    An updating rule that systematically distorts beliefs is strictly Blackwell monotone if and only if it is Bayes' law.
\end{restatable}
That is, an updating rule is strictly Blackwell monotone if and only if \(\varphi(\mu) = \mu\) for all \(\mu \in \Delta\). The theorem follows from a number of smaller results, the proofs to which may be found in Appendix \ref{maintheoremproof}. We begin by arguing that if an updating rule is Blackwell monotone, the distortion \(\varphi\) must be continuous on \(\Delta^{\circ}\).
\begin{restatable}{lemma}{continuitylemma}
   \label{continuitylemma}
    If updating rule \(U\) is Blackwell monotone, \(\varphi\) is continuous on \(\Delta^{\circ}\). 
\end{restatable}
The intuition behind this lemma is straightforward. \(U\) being Blackwell monotone is equivalent to the induced value function \(W\)'s convexity in the Bayesian posterior \(W\). But convex functions are continuous except on the boundaries of their domain. Thus, if \(\varphi\) ``jumps'' somewhere, we can find a decision problem in which the ``jump'' is inherited by \(W\), rendering it non-convex.

We then build upon this lemma to discipline the behavior of \(\varphi\) for all full-support beliefs. We say that a distortion \(\varphi\) is \textcolor{OrangeRed}{Trivial} on a set \(Y \subseteq \Delta\) if \(\varphi(\mu) = \bar{\mu}\) for some \(\bar{\mu} \in \Delta\) for all \(\mu \in Y\). \(\varphi\) is the \textcolor{OrangeRed}{Identity} map on \(Y \subseteq \Delta\) if \(\varphi(\mu) = \mu\) for all \(\mu \in Y\). We also highlight one particular departure from Bayesian updating:
\begin{definition}
    Let there be two states (\(n = 2\)). A DM with updating rule \(U\) displays \textcolor{OrangeRed}{Extreme-Belief Aversion} if there exist two intervals \(C_1 \coloneqq (0,c)\) and \(C_2 \coloneqq (d,1)\) (with \(c \leq d\)) such that \(\varphi(\mu) = c\) for all \(\mu \in C_1\),  \(\varphi(\mu) = d\) for all \(\mu \in C_2\)  and \(\varphi(\mu) = \mu\) for all \(\mu \in \left[c,d\right]\).
\end{definition}
We discuss this pattern of behavior in Appendix \ref{appx}, where we specialize to the two-state environment.
\begin{restatable}{proposition}{propxba}
    \label{propxba}
\begin{enumerate}[label={(\roman*)},noitemsep,topsep=0pt]
    \item Let there be two states (\(n=2\)). If \(U\) is Blackwell monotone, the DM displays extreme-belief aversion.
    \item Let there be three or more states (\(n \geq 3\)). If \(U\) is Blackwell monotone, \(\varphi\) is either trivial or the identity map on \(\Delta^{\circ}\).
\end{enumerate}
\end{restatable}
Lemma \ref{continuitylemma} is key to establishing this proposition, whose proof is purely topological. Thanks to the lemma, we know that Blackwell monotonicity implies continuity of the distortion \(\varphi\) on the interior of the simplex. Using this continuity, we argue that if \(\varphi\) is not the identity map on \(\Delta^{\circ}\) it must be locally trivial at some belief \(\mu \in \Delta^{\circ}\) for which \(\varphi(\mu) \neq \mu\). That is, every belief in an open ball around that \(\mu\) must be mapped to the same point. This local triviality plus the facts that \(\varphi\) is continuous and \(\Delta^{\circ}\) is connected allow us to conclude the global structure stated in the proposition.\footnote{I am grateful to the anonymous reviewer who suggested this approach.} Moreover, if we add the additional requirement that \(\varphi\) is continuous on \(\Delta\), we arrive at the following stark conclusion when there are three or more states.
\begin{corollary}\label{nontrivial}
    Let there be three or more states (\(n \geq 3\)). If \(U\) is Blackwell monotone and \(\varphi\) is continuous on \(\Delta\), \(\varphi\) is either trivial or the identity map on \(\Delta\).
\end{corollary}
But \(U\) cannot be Blackwell monotone and such that the DM gets every full-support posterior ``right,'' yet gets some non-full-support posterior ``wrong:''
\begin{restatable}{lemma}{identityeverywhere}
    \label{identityeverywhere}
    If \(U\) is Blackwell monotone and \(\varphi\) is the identity map on \(\Delta^{\circ}\), \(\varphi\) is the identity map on \(\Delta\).
\end{restatable}
The reasoning behind this lemma is simple. If \(\varphi\) distorts some belief, \(x\), on the boundary of \(\Delta\) but none on the interior, we can take a binary experiment \(\pi\) whose induced Bayesian distribution over posteriors has one of its two support points on that \(x\). But then we just take an ever-so-slightly less informative experiment, \(\pi'\), in which \(x\) is replaced with \(x'\) that is on the line segment between \(x\) and the prior. No matter how close it is to \(x\), \(x'\) lies, nevertheless, in the interior of \(\Delta\), and so \(\varphi(x') = x'\). Consequently, we can always find a decision problem in which \(x\) leads to a harmful mistake but \(x'\) does not, or a decision problem in which the DM fails to take advantage of her information at \(x\) but does at \(x'\).

On the other hand, if \(\varphi\) is trivial on some open ball in \(\Delta^{\circ}\), the updating rule cannot be strictly Blackwell monotone:
\begin{restatable}{lemma}{fullcollapse}
    \label{lemma37}
    If there exists \(\mu' \in \Delta^{\circ}\) and \(\varepsilon > 0\) such that \(\varphi(\mu) = \bar{\mu} \in \Delta\) for all \(\mu \in B_{\varepsilon}(\mu')\), \(U\) is not strictly Blackwell monotone.
\end{restatable}
These results combine to yield the theorem. Even without the additional strictness refinement, if we impose that \(\varphi\) is continuous and nontrivial, Blackwell monotonicity already singles out Bayes' law when there are more than three states (Corollary \ref{nontrivial}). Moreover, the continuity assumption only has bite for extreme posteriors--the distortion is necessarily continuous for full-support posteriors. And the two-state environment? We show in Appendix \ref{appx} that a continuous updating rule is Blackwell monotone if and only if the DM displays extreme-belief aversion.

Before discussing further the intuition behind Theorem \ref{maintheorem}, we turn our attention to updating rules that yield a positive value of information from the agent's perspective. We say that distortion \(\varphi \colon \Delta \to \Delta\) is affine if \(\varphi\left(\mu\right) = A \mu + b\) for some invertible matrix \(A\) and vector \(b\). Then,
\begin{restatable}{theorem}{secondmaintheorem}
   \label{secondmaintheorem}
    An updating rule that systematically distorts beliefs is strictly non-paternalistically Blackwell monotone if and only if the distortion is affine.
\end{restatable}
Unlike the paternalistic computation of the DM's welfare, where we use the Bayesian posterior, we evaluate her non-paternalistic welfare according to the actual posteriors the DM uses to make her decisions. This means that an updating rule is non-paternalistically Blackwell monotone if the DM's distortion is convexity-preserving: \(V \circ \varphi\) must be convex for any convex function \(V\).

We show first that non-paternalistic Blackwell-monotonicity is equivalent to the distortion taking the form \(\varphi\left(\mu\right) = A \mu + b\) for some \(\left(n-1\right) \times \left(n-1\right)\) matrix \(A\) and \(b \in \mathbb{R}^{n-1}\). That is, the distortion consists of a belief-independent scaling and translation of the posterior. Only such persistent biases preserve the convexity of any \(V\) and, hence, a positive value of information. Strict non-paternalistic Blackwell-monotonicity also requires that \(\varphi\) be injective, which is equivalent to matrix \(A\) being of full rank. This is intuitive: if two Bayesian beliefs were distorted to the same point, so too must any convex combination of the two beliefs. But then there are informational improvements that never have strictly positive value; and, consequently, the updating rule is not strictly non-paternalistically Blackwell monotone. We defer the proof of the theorem to Appendix \ref{secondmaintheoremproof}.

\subsection{Sketching a Special Setting}\label{sketchysec}

To gain intuition for Theorem \ref{maintheorem}, let us sketch an alternative proof for a restricted class of updating rules. We say a distortion, \(\varphi\), is \textcolor{OrangeRed}{Grounded} if \(\varphi(\mu_0) = \mu_0\). 

The result to be argued is
\begin{corollary}\label{updatecorr}
    An updating rule whose distortion is grounded is strictly Blackwell monotone if and only if it is Bayes' law.
\end{corollary}
Naturally, we need only establish the necessity of Bayes' law. There are three key ingredients for this result:

\bigskip

\noindent \textbf{1. The Necessity of Underreaction.} From discussion in \cite{braghieri2023biased} and as illustrated in \S\ref{ynot}, if \(\varphi\) is grounded, \(U\) is Blackwell monotone only if \(\varphi(\mu) \in \conv\left\{\mu_0,\mu\right\}\) for any \(\mu \in \Delta\).\footnote{\(\conv X\) denotes the closed convex hull of set \(X\).} Just as is depicted in Figure \ref{fig1}, the existence of any other error means that there is a (binary) decision problem where the DM is strictly worse off from some information than if she had obtained no information.

\bigskip

\noindent \textbf{2. The Necessity of Monotonicity.} Our previous observation allows us to restrict attention without loss of generality to the two-state environment, which we now do. Next, we deduce that if \(\varphi\) is grounded, a Blackwell-monotone \(U\) must distort posteriors in a monotone way: if \(\mu_L < \mu_{L}' < \mu_0\), then \(\hat{\mu}_L \leq \hat{\mu}_{L}' \leq \mu_0\). Indeed, suppose for the sake of contradiction not, i.e., that \(\hat{\mu}_{L}' < \hat{\mu}_L\). As the DM is underreacting, we must have \(\mu_{L}' \leq \hat{\mu}_{L}' < \hat{\mu}_L\), which means that \(\conv\left\{\mu_{L}', \hat{\mu}_{L}'\right\}\) can be strictly separated from \(\hat{\mu}_L\) by some hyperplane. Consequently, there is a decision problem in which the DM takes fails to take advantage of information only with the more informative of two experiments, violating Blackwell monotonicity.

\bigskip

\noindent \textbf{3. The Necessity of Triviality.} Finally, if \(\varphi\) is grounded, a Blackwell-monotone \(U\) in which there is underreaction must be such that an interval of (Bayesian) beliefs must all be mapped to the same belief by \(\varphi\). The logic behind this is precisely that given in \S\ref{ynot}: if we did not have this constancy of \(\varphi\) on an interval, we could strictly improve the DM's payoff by moving from a more informative ternary distribution over posteriors to a binary one, as depicted in Figure \ref{fig2}.

But this local triviality of \(\varphi\) means that certain local increases in information can never be exploited, violating strict Blackwell monotonicity. We conclude that if \(\varphi\) is grounded, a strictly Blackwell-monotone \(U\) cannot produce any errors. It must be Bayes' law.

\section{Non-Systematic Distortions}\label{nonsys}

Now let us remove the requirement that the updating rule systematically distorts beliefs. That is, given experiment \(\left(\pi,S\right) \in \Pi\), and upon observing realization \(s \in S\), the DM's posterior is 
\[\hat{\mu}_s\left(\theta; \pi\right) = \varphi_{\pi}\left(\mu_s\left(\theta; \pi\right)\right)\text{,}\] 
i.e., the distortion of the Bayesian posterior depends on what the experiment was itself.\footnote{We do not permit the distortion to depend on the \textit{realization} \(s\). This is without loss of generality as we have assumed that each realization induces a unique Bayesian posterior. Allowing for duplicate Bayesian beliefs and realization-contingent distortions would not affect our results.}  We maintain the standing consistency requirement on the DM's interim choice of action. Then, for a fixed statistical experiment \((\pi,S)\), decision problem, and consistent choice of action, the function \(W_\pi\left(\mu\right) \coloneqq \mathbb{E}_{\mu}u\left(a^*\left(\varphi_\pi\left(\mu\right)\right),\theta\right)\) is a well-defined function of the Bayesian posterior \(\mu\) on \(\supp \rho\).

In this section, we consider three additional classes of updating rules, which we term grounded, convex, and focused. For any pair of experiments \(\left(\pi, S\right)\) and \(\left(\pi',S'\right)\), a \textcolor{OrangeRed}{Convex Combination} of the two experiments is, for some \(\lambda \in \left[0,1\right]\), experiment \(\left(\pi_{\lambda}, S \cup S'\right)\), where \(\pi_{\lambda}\left(\left.s\right|\theta\right) = \lambda \pi \left(\left.s\right|\theta\right)\) for all \(s \in S\), for all \(\theta \in \Theta\) and \(\pi_{\lambda}\left(\left.s'\right|\theta\right) = (1-\lambda) \pi \left(\left.s'\right|\theta\right)\) for all \(s' \in S'\), for all \(\theta \in \Theta\). Our definition of grounded updating rules is a simple modification of the previous term concerning distortions. \textit{Viz.}, we say an updating rule is \textcolor{OrangeRed}{Grounded} if the DM's posterior upon observing the trivial experiment with a single signal realization is the prior. We say an updating rule is \textcolor{OrangeRed}{Focused} if for any pair of experiments \(\left(\pi, S\right)\) and \(\left(\pi',S'\right)\), if \(\pi(\left.s\right|\theta) = \pi'(\left.s'\right|\theta)\) for all \(\theta \in \Theta\) for some pair \(s \in S\), \(s' \in S'\), then \(\hat{\mu}_s\left(\theta; \pi\right) = \hat{\mu}_{s'}\left(\theta; \pi'\right)\).
\begin{definition}
    An Updating Rule is \textcolor{OrangeRed}{Convex} if for any triple of experiments \(\left(\pi, S\right)\), \(\left(\pi',S'\right)\), and \(\left(\pi_{\lambda}, S \cup S'\right)\), where \(\left(\pi_{\lambda}, S \cup S'\right)\) is a convex combination of the first two experiments, and consistent choice of action \(a^* \colon \Delta \to A\),
\[\lambda \mathbb{E}_{\rho} W_\pi(\mu) + (1-\lambda) \mathbb{E}_{\rho'} W_{\pi'}(\mu) \geq \mathbb{E}_{\rho_{\lambda}} W_{\pi_{\lambda}}(\mu)\text{,}\]
where \(\rho\), \(\rho'\) and \(\rho_{\lambda} = \lambda \rho + (1-\lambda) \rho'\) are the corresponding distributions over Bayesian posteriors (respectively).
\end{definition}
Any updating rule that systematically distorts beliefs is both convex and focused, and so most updating rules studied in the literature satisfy these properties. An updating rule that does not systematically distort beliefs, yet is convex, is ``no learning without full disclosure,'' where for any experiment other than the fully informative one, each realization is updated to the prior, and the DM updates the fully-informative experiment according to Bayes law. In fact, this updating rule is both Blackwell monotone and grounded, but it is not strictly Blackwell monotone.

Focused updating rules possess a weaker form of column-independence than those that systematically distort beliefs. The latter class are those for which the DM must be able to understand when two ratios are equivalent, irrespective of the precise numerators and denominators of the fractions. \textit{Viz.}, a systematic distorter understands that \(1/2\) and \(9/18\) are equivalent. Focused updating rules are those for which the DM must be able to understand when two \textit{identical fractions} are the same--i.e., that \(1\) equals \(1\), \(2 = 2\), \(9=9\) and \(18=18\)--but may not recognize that \(1/2 = 9/18\). Accordingly focused updating rules are a much broader class, and assume a much lower baseline of cognition.\footnote{To be fair, it \textit{is} more difficult to recognize that \(.28/.95\) and \(7/19\) are the same (are they?) than \(7/19\) and \(7/19\).} For such updaters, focus merely imposes that upon seeing the exact same thing \textit{ex interim}, a DM's beliefs must be the same.

The primary motivation behind each of these properties is normative. Groundedness and focus are both forms of consistency. If the agent learns nothing and it is obvious that she has learned nothing--the experiment is completely uninformative, with just one signal realization--her belief should not change. If an agent observes signal realization \(s\) and is told the conditional probabilities \(\left(\pi\left(\left.s\right|\theta\right)\right)_{\theta \in \Theta}\), she should not hold different beliefs depending on what the probabilities of counterfactual signal realizations are. Convexity corresponds to the notion that simpler experiments ought to lead to fewer and milder mistakes than more complicated experiments, and a convex combination of experiments is more complex than each component experiment in a natural sense.

\subsection{A Strictly Monotone Updating Rule Other Than Bayes' Law}\label{nonbayesex}

Absent the assumption of systematic distortions, groundedness, by itself, is not enough to single out Bayes' law as the unique Blackwell-monotone updating rule, nor even the unique strictly Blackwell-monotone rule. Let us construct an alternative one. Let \(\Theta = \left\{0,1\right\}\). For any experiment \(\left(\pi,S\right)\), with corresponding \(\rho\), define
\[\sigma(\pi) \coloneqq \mathbb{E}_\rho\left[\frac{\left(\mu-\mu_0\right)^{2}}{\mu_0\left(1-\mu_0\right)}\right]\text{,}\]
\[\ubar{\mu}_\pi \coloneqq \mu_0 (1- \sigma(\pi)), \quad \text{and} \quad \bar{\mu}_\pi \coloneqq \mu_0 + \sigma(\pi) (1-\mu_0)\text{.}\]
\(\sigma \colon \Pi \to \left[0,1\right]\) is strictly monotone in the Blackwell order: if \(\pi \succ \pi'\), \(\sigma\left(\pi\right) > \sigma\left(\pi'\right)\). \(\sigma(\pi)\) also ranges from \(0\) (a fully uninformative experiment) to \(1\) (full information).

We define
\[\varphi_\pi(\mu) = \begin{cases}
    \ubar{\mu}_\pi, \quad &\text{if} \quad \mu \leq \ubar{\mu}_\pi\\
    \mu, \quad &\text{if} \quad \ubar{\mu}_\pi < \mu < \bar{\mu}_\pi\\
    \bar{\mu}_\pi, \quad &\text{if} \quad \mu \geq \bar{\mu}_\pi\text{.}
\end{cases}\]
That is, \(\varphi_\pi\) corresponds to experiment-dependent extreme-belief aversion, where the degree to which the DM eschews obtaining extreme beliefs is strictly decreasing in the informativeness of the experiment. Observe that this updating rule is grounded. Moreover,
\begin{remark}\label{remarkremark}
    This updating rule is strictly Blackwell monotone.
\end{remark}

Please visit Appendix \ref{remarkremarkproof} for a proof of this remark. This rule alters the standard extreme-belief-aversion rule, which is Blackwell monotone (proved in Appendix \ref{appx})--but not strictly Blackwell monotone--to make it strictly so via the experiment-dependent weight \(\sigma\left(\cdot\right)\). But can there be strictly Blackwell-monotone rules other than Bayes' law when there are three or more states?

Yes, there can: for example, suppose that for each line that goes through \(\mu_0\) the DM updates any experiment whose Bayesian posteriors all lie on that line in the manner defined in Remark \ref{remarkremark}, treating the intersection of the simplex and the specified line as the two-state \(1\)-simplex. Any experiment whose Bayesian posteriors are not all collinear is updated according to Bayes' law. This is a strictly Blackwell-monotone (and grounded) updating rule.

\subsection{Convex Updating}

Nevertheless, if we require that the updating rule be not only grounded but convex, Bayes' law is the only strictly Blackwell-monotone updating rule.

\begin{restatable}{proposition}{convexprop}
    \label{convexprop}
    A grounded, convex updating rule is strictly Blackwell monotone if and only if it is Bayes' law.
\end{restatable}
Only the necessity direction needs proving, as Bayes' law is grounded, strictly Blackwell monotone, and convex. We defer the detailed proof to Appendix \ref{convexpropproof} and discuss a sketch here.

Our approach mimics that for Corollary \ref{updatecorr}. We begin by showing that for binary experiments, convexity and Blackwell monotonicity imply Bayesian updating. First, via the exact same logic as in the proof of Corollary \ref{updatecorr}, we argue that Blackwell-monotonicity implies that any error must be an \textbf{underreaction} to information. Otherwise, just as in \S\ref{ynot}, we could find a decision problem in which the DM strictly prefers no information to some. Second, if it is Blackwell monotone, the updating rule must also distort posteriors in a \textbf{monotone} manner: for two binary experiments \(\pi\) and \(\pi'\), whose Bayesian posteriors are \(\left\{\mu_{L},\mu_H\right\}\) and \(\left\{\mu_{L}',\mu_H\right\}\) (\(\mu_L < \mu_L' < \mu_0 < \mu_H\)), it must be that \(\varphi_\pi(\mu_L) \leq \varphi_{\pi'}(\mu_{L}')\).

Third, we show that Blackwell-monotonicity and \textit{convexity} imply that the updating rule must exhibit a form of \textbf{triviality}. If the updating rule is not Bayes' law for some binary experiment but Blackwell monotone and convex, there necessarily exists a binary experiment \(\pi_1\), whose Bayesian posteriors are \(\left\{\mu_{l}',\mu_h\right\}\), and a ternary experiment \(\pi_2\), whose Bayesian posteriors are \(\left\{\mu_{l},\mu_m, \mu_h\right\}\) (\(\mu_l < \mu_l' < \mu_m < \mu_0 < \mu_h\)) such that \(\pi_2\) is more informative than \(\pi_1\), \(\varphi_{\pi_2}(\mu_l) = \varphi_{\pi_1}(\mu_l') = \varphi_{\pi_2}(\mu_m)\), and \(\varphi_{\pi_1}(\mu_h) = \varphi_{\pi_2}(\mu_h)\). This step is quite involved as it requires constructing a number of auxiliary binary experiments, two of which are such that \(\pi_2\) is a convex combination of them, which is where the convexity of the updating rule comes into play. Namely, convexity allows us to conclude that \(\varphi_{\pi_2}(\mu_l) = \varphi_{\pi_2}(\mu_m)\)--and so both equal \(\varphi_{\pi_1}(\mu_l')\), which must be sandwiched between them due to monotonicity--and \(\varphi_{\pi_1}(\mu_h) = \varphi_{\pi_2}(\mu_h)\). This triviality is incompatible with strict Blackwell monotonicity, so we conclude that the only convex and strictly Blackwell-monotone updating rule for binary experiments is Bayes' law.

It remains only to deduce the necessity of Bayes' law for general, not-necessarily-binary experiments. But this follows easily from the binary-experiment result: for an arbitrary experiment in which some realization produces an error, we can find a (less-informative) binary experiment and a decision problem that would yield the same payoff to a Bayesian under each experiment, but for which the updating error leads to a strict payoff decrease. This contradicts Blackwell monotonicity.

\subsection{Focused Updating}

Adding focus to groundedness is also enough to single out Bayes' law as the only strictly Blackwell-monotone updating rule.

\begin{restatable}{proposition}{focusprop}
    \label{focusprop}
    A grounded, focused updating rule is strictly Blackwell monotone if and only if it is Bayes' law.
\end{restatable}
The proof of this proposition is extremely similar to that of Proposition \ref{convexprop}. In particular, that Blackwell monotonicity implies that any error must be an underreaction and that the updating rule must distort posteriors produced by binary experiments in a monotone way hold through the exact same logic. Likewise, the final step--going from the binary-experiment conclusion to arbitrary experiments is identical as well. The lone difference is in establishing that the same variety of local triviality must also manifest. 

In this step, now focus plays a role. In particular, like with the convexity result, we use a collection of carefully-constructed auxiliary binary experiments to show that underreaction to some binary experiment, focus, and Blackwell monotonicity necessitate the existence of a ternary experiment and a less informative experiment that render the same value of information to the DM no matter the decision problem, violating strict Blackwell monotonicity. To elaborate, in deriving the convex-updating rule result, the relevant ternary experiment is a convex combination of two of the auxiliary experiments. Here, we go from an auxiliary ternary experiment to a binary experiment by combining two of the realizations--which strictly decreases the experiment's informativeness--but keeping the third unchanged. By focus, that holdover realization must produce the same belief in each of the two experiments, which is enough (with a little more effort) to generate the desired local triviality that is incompatible with strict Blackwell monotonicity. For the details, please visit Appendix \ref{focuspropproof}.

\appendix

\section{Omitted Proofs}\label{theappendix}

\subsection{Theorem \ref{maintheorem} Proof}\label{maintheoremproof}
\maintheorem*
We prove this through a sequence of smaller results.

\continuitylemma*
\begin{proof}
    To that end, we prove the contrapositive: if \(\varphi\) is not continuous on \(\Delta^{\circ}\), there is a decision problem in which the induced \(W\) is not convex on \(\Delta^{\circ}\). By Theorem 10.1 in \cite{rockafellar1970convex} it suffices to show that \(W\) is not continuous on \(\Delta^{\circ}\). 
    
    We enumerate the states \(\Theta = \left\{0, 1, \dots, n-1\right\}\) and for \(\mu \in \Delta\) let \(\mu^i\) denote \(\mathbb{P}\left(i\right)\) (\(i = 1, \dots, n-1\)). Suppose \(\varphi\) is not continuous at some \(\mu \in \Delta^{\circ}\). That is, there exists some sequence \(\left\{\mu_n\right\} \subseteq \Delta^{\circ}\) such that \(\mu_n \to \mu\) but \(\varphi(\mu_n) \not\to \varphi(\mu)\). Without loss of generality (as we could just relabel) we assume that \(\varphi^1(\mu_n) \not\to \varphi^1(\mu)\).
    
    Now, let us construct a decision problem as follows. The set of actions is \(A \equiv \left[0,1\right]\), and the DM's utility function is
    \[u\left(a,\theta\right) = -\mathbf{1}_{\left\{\theta = 1\right\}} (a-1)^2 - \left(1-\mathbf{1}_{\left\{\theta = 1\right\}}\right) a^2\text{.}\]
    Consequently, given belief \(\mu \in \Delta\), the DM's interim expected payoff (which guides her decision) is
    \[-\mu^1 \left(a - 1\right)^2 - \left(1-\mu^1\right)a^2\text{.}\]
    By construction, \(a^*\left(\varphi(\mu)\right) = \varphi^1\left(\mu\right)\), so 
    \[W(\mu) = - \mu^1 \left(\varphi^1\left(\mu\right) - 1\right)^2 - \left(1-\mu^1\right)\varphi^1\left(\mu\right)^2 = -\mu^1 + 2 \mu^1 \varphi^1(\mu) - \varphi^1(\mu)^2\]
    
    Observe that \(W(\mu_n) \to W(\mu)\) if and only if
    \[\varphi^1(\mu_n) (2 \mu^1_n - \varphi^1(\mu_n)) - \varphi^1(\mu)(2 \mu^1 - \varphi^1(\mu)) \to 0\text{,}\]
    which holds if and only if \(\varphi^1(\mu_n) \to \varphi^1(\mu)\) or \(\varphi^1(\mu_n) \to 2 \mu^1 - \varphi^1(\mu)\). 
    
    We have assumed away the first possibility, so if \(\varphi^1(\mu_n) \not\to 2 \mu^1 - \varphi^1(\mu)\), we are done. Accordingly, let \[\varphi^1(\mu_n) \to \tilde{\mu}^1 \coloneqq 2 \mu^1 - \varphi^1(\mu) \neq \varphi^1(\mu)\text{.}\] Observe that this implies \(\varphi^1(\mu) \neq \mu^1 \neq \tilde{\mu}^1\). Because of this, there exist scalars \(\alpha, \beta \in \mathbb{R}\) such that \[\alpha \varphi^1(\mu) + \beta > 0, \quad \alpha \mu^1 + \beta > 0, \quad \text{and} \quad 0 > \alpha \tilde{\mu}^1 + \beta\text{,}\]
    or there exist scalars \(\alpha, \beta \in \mathbb{R}\) such that
    \[\alpha \varphi^1(\mu) + \beta > 0, \quad \alpha \mu^1 + \beta < 0, \quad \text{and} \quad 0 > \alpha \tilde{\mu}^1 + \beta\text{.}\]

    Suppose the first case is feasible. We define a new two-action decision problem with action set \(\left\{a_1,a_2\right\}\), where the payoff to action \(a_2\) is the state-independent \(0\) and the expected payoff to action \(a_1\) (in belief \(\mu\)) is \(\alpha \mu^1 + \beta\). Thus, recalling that \(W(\mu) \coloneqq \mathbb{E}_\mu u\left(a^*\left(\varphi\left(\mu\right)\right),\theta\right)\),
    \[W(\mu_n) \to 0 < \alpha \mu^1 + \beta = W(\mu)\text{,}\]
    and so \(W\) is not continuous at \(\mu\). We omit the second case as it mirrors the first.\end{proof}

    \propxba*
    
\begin{proof}
I'm grateful to an anonymous referee for sketching this approach. Let \(U\) be Blackwell monotone. By Lemma \ref{continuitylemma}, this implies that \(\varphi\) is continuous on \(\Delta^{\circ}\). Suppose there exists some \(x \in \Delta^{\circ}\) for which \(\varphi\left(x\right) = y \neq x\).
\begin{claim}
There exists an \(\varepsilon > 0\) such that for all \(x' \in B_{\varepsilon}(x)\), \(\varphi\left(x'\right) = y\).
\end{claim}
\begin{proof}
Suppose for the sake of contradiction not. Then there exists a sequence \(\left\{x_n\right\}\) in \(\Delta\) that converges to \(x\) and such that \(y_n \coloneqq \varphi\left(x_n\right) \neq y\) for all \(n \in \mathbb{N}\).

Consider a two-action decision problem in which the payoff (in belief \(\mu\)) to action \(1\) is \(0\) and the payoff to action \(2\) is \(\alpha \mu - \beta\), where \(\alpha \in \mathbb{R}^{n-1}\) and \(\beta \in \mathbb{R}\) are such that \(\alpha y - \beta = 0\) and \(\alpha x - \beta < 0\). We must have i. \(\alpha y_n - \beta \geq 0\) for infinitely many members of the sequence \(\left\{y_n\right\}\); or ii. \(\alpha y_n - \beta < 0\) for infinitely many members of the sequence \(\left\{y_n\right\}\).

i. In the first case, by construction, there is a subsequence \(\left\{y_{n_k}\right\}\) such that \(\alpha y_{n_k} - \beta \geq 0\) for all \(y_{n_k}\). We impose for all such beliefs, the DM takes action \(2\) and for belief \(y\), the DM takes action \(1\). As \(x_n \to x\), \(x_{n_k} \to x\), so 
\[\lim_{n_k \to \infty} W\left(x_{n_k}\right) = \lim_{n_k \to \infty} \alpha x_{n_k} - \beta = \alpha x - \beta < 0 = W(x)\text{.}\]

ii. In the second case, by construction, there is a subsequence \(\left\{y_{n_k}\right\}\) such that \(\alpha y_{n_k} - \beta < 0\) for all \(y_{n_k}\), so for all such beliefs, the DM takes action \(1\). Accordingly, \(W(x_{n_k}) = 0\) for all \(x_{n_k}\). We impose for belief \(y\), the DM takes action \(2\), so 
\[\lim_{n_k \to \infty} W\left(x_{n_k}\right) = 0 > \alpha x - \beta = W(x)\text{.}\]
In both cases, \(W\) is discontinuous at \(x\) and, therefore, non-convex, contradicting that \(U\) is Blackwell monotone. \end{proof}
This claim implies that either \(\varphi^{-1}\left(y\right)\) is open in \(\left(\Delta^{\circ}, \left\|\cdot\right\|_{E}\right)\) or is the union of an open set with \(y\) itself (in the event where \(\varphi\left(y\right) = y\) as well). By the continuity of \(\varphi\), \(\varphi^{-1}\left(y\right)\) is closed in \(\left(\Delta^{\circ}, \left\|\cdot\right\|_{E}\right)\). As \(\Delta^{\circ}\) is connected, if there are three or more states \(\varphi^{-1}(y) = \Delta^{\circ}\), i.e., \(\varphi\) is trivial on \(\Delta^{\circ}\). If \(n = 2\), either \(\varphi^{-1}(y) = \left(0,1\right)\), \(\varphi^{-1}(y) = \left(0,c\right]\), or \(\varphi^{-1}(y) = \left[d,1\right)\). \end{proof} 

\identityeverywhere*
\begin{proof}
    Suppose for the sake of contraposition that there exists some \(x \in \Delta^{\circ} \setminus \Delta\) for which \(\varphi(x) \neq x\). If \(\hat{x}\) does not lie on the line segment between \(x\) and \(\mu_0\), we are done, as there exists a decision problem for which the DM strictly prefers no information to any binary experiment whose distribution over Bayesian posteriors has support on \(x\) (via the strict separating hyperplane theorem).
    
    Thus, for some \(\lambda \in \left[0,1\right)\), \(\hat{x} = \lambda x + (1-\lambda) \mu_0\). But then there exists a decision problem for which a DM prefers the binary experiment that induces \(\rho'\) with support \(\left\{y, \eta \hat{x} + (1-\eta) x\right\}\) to that inducing \(\rho\) with support \(\left\{y,x\right\}\). Accordingly, \(U\) is not Blackwell monotone.\end{proof}

\fullcollapse*
    \begin{proof}
    Suppose that \(\varphi\) is locally trivial, i.e., assume the stated condition. Take a binary experiment, \(\pi\), whose induced Bayesian distribution over posteriors, \(\rho\), is supported on \(\left\{\mu',\mu''\right\}\) for some \(\mu'' \in \Delta\). Then, take another experiment, \(\tilde{\pi}\), whose Bayesian distribution over posteriors, \(\tilde{\rho}\) has support on \(\left\{\mu_1, \mu_2, \mu''\right\}\), where \(\mu_1, \mu_2 \in B_{\varepsilon}(\mu')\) and \(\mathbb{P}_{\rho}(\mu'') = \mathbb{P}_{\tilde{\rho}}(\mu'')\). By construction, \(\tilde{\pi} \succ \pi\), yet for any decision problem and any consistent choice of action, \(\pi\) and \(\tilde{\pi}\) yield the DM the same \textit{ex ante} expected utility, violating strict Blackwell monotonicity. \end{proof}

    Finally, we combine the intermediate results to produce the theorem.
    \begin{proof}[Proof of Theorem \ref{maintheorem}.]
        Proposition \ref{propxba} and Lemma \ref{lemma37} tell us that if \(U\) is strictly Blackwell monotone, \(\varphi\) must be the identity map on \(\Delta^{\circ}\). Lemma \ref{identityeverywhere} states that this property must extend to the boundary of \(\Delta\) as well, which yields the theorem.\end{proof}
        
    \subsection{Theorem \ref{secondmaintheorem} Proof}\label{secondmaintheoremproof}
\secondmaintheorem*

Recall that an updating rule is non-paternalistically Blackwell monotone if and only if in any decision problem \(V(\varphi(\mu))\) is convex.
    \begin{lemma}
        \(V\left(\varphi\left(x\right)\right)\) is convex for all convex \(V\) if and only if \(\varphi(x) = A x + b\) for some \(\left(n-1\right) \times \left(n-1\right)\) matrix \(A\) and \(b \in \mathbb{R}^{n-1}\).
    \end{lemma}
    \begin{proof}
        \noindent \(\left(\Rightarrow\right)\) If \(\varphi\left(x\right) = A x + b\), then for all \(x, x' \in \Delta\) and \(\lambda \in \left(0,1\right)\)
        \[\begin{split}
        V\left(\varphi\left(\lambda x + \left(1-\lambda\right) x'\right)\right) &= V\left(\lambda \left(A x + b\right) + \left(1-\lambda\right)\left(A x' + b\right)\right)\\ &\leq \lambda V\left(A x + b\right) + \left(1-\lambda\right)V\left(A x' + b\right)\\ &= \lambda V\left(\varphi\left(x\right)\right) + \left(1-\lambda\right)V\left(\varphi\left(x'\right)\right)\text{ ,}
        \end{split}\]
        so \(V \circ \varphi\) is convex.

        \noindent \(\left(\Leftarrow\right)\) Suppose for the sake of contraposition that there exist distinct \(x, x' \in \Delta\) and \(\lambda \in \left(0,1\right)\) such that \(\varphi\left(x\right) = A x + b\) and \(\varphi\left(x'\right) = A x' + b\) but \[\tag{\(B.1\)}\label{eqB1}\varphi\left(\lambda x + \left(1-\lambda\right) x'\right) \neq A\left(\lambda x + \left(1-\lambda\right) x'\right) + b\text{.}\]
        Let \(V\left(x\right) = \alpha x\), where \(\alpha \in \mathbb{R}^{n-1}\), so
        \[\tag{\(B.2\)}\label{eqB2}\begin{split}
        \lambda V\left(\varphi\left(x\right)\right) + \left(1-\lambda\right)V\left(\varphi\left(x'\right)\right) &= \lambda \alpha \varphi\left(x\right) + \left(1-\lambda\right) \alpha \varphi\left(x'\right)\\
        &= \lambda \alpha A x + \left(1-\lambda\right) \alpha A x' + \alpha b\\
        &= \alpha A\left(\lambda x + \left(1-\lambda\right) x'\right) + \alpha b
        \text{ .}
        \end{split}\]
        Appealing to Expression \ref{eqB1}, WLOG we assume \[\alpha \left(A\left(\lambda x + \left(1-\lambda\right) x'\right) + b\right) \neq \alpha \varphi\left(\lambda x + \left(1-\lambda\right) x'\right) \] (as otherwise we could just modify \(\alpha\)). If 
        \[\alpha \left(A\left(\lambda x + \left(1-\lambda\right) x'\right) + b\right) < \alpha \varphi\left(\lambda x + \left(1-\lambda\right) x'\right)\text{,}\]
        we have, from Equation \ref{eqB2},
        \[\begin{split}
        \lambda V\left(\varphi\left(x\right)\right) + \left(1-\lambda\right)V\left(\varphi\left(x'\right)\right)
        &= \alpha \left(A\left(\lambda x + \left(1-\lambda\right) x'\right) + b\right)\\
        &< \alpha \varphi\left(\lambda x + \left(1-\lambda\right) x'\right) = V\left(\varphi\left(\lambda x + \left(1-\lambda\right) x'\right)\right)
        \text{.}
        \end{split}\]
        so \(V \circ \varphi\) is not convex. If 
        \[\alpha \left(A\left(\lambda x + \left(1-\lambda\right) x'\right) + b\right) > \alpha \varphi\left(\lambda x + \left(1-\lambda\right) x'\right)\text{,}\]
        we simply define \(V\left(x\right) = - \alpha x\), in which case, again, we have 
        \[\lambda V\left(\varphi\left(x\right)\right) + \left(1-\lambda\right)V\left(\varphi\left(x'\right)\right) < V\left(\varphi\left(\lambda x + \left(1-\lambda\right) x'\right)\right)
        \text{,}\]
        so \(V \circ \varphi\) is not convex. \end{proof}
    We say that \(\varphi(x) = A x + b\) is affine if \(A\) is invertible. It is straightforward to finish proving the theorem. \begin{proof}[Proof of Theorem \ref{secondmaintheorem}.]
        \(\left(\Leftarrow\right)\) It suffices to show that there always exists a strictly convex \(V\) such that \(V \circ \varphi\) is strictly convex. If \(\varphi\left(x\right) = A x + b\) for some invertible \(A\), then \(\varphi\) is injective. Consequently, for all distinct \(x, x' \in \Delta\) and \(\lambda \in \left(0,1\right)\), \(A x \neq A x'\), so 
        \[\begin{split}
        V\left(\varphi\left(\lambda x + \left(1-\lambda\right) x'\right)\right) &= V\left(\lambda \left(A x + b\right) + \left(1-\lambda\right)\left(A x' + b\right)\right)\\ &< \lambda V\left(A x + b\right) + \left(1-\lambda\right)V\left(A x' + b\right)\\ &= \lambda V\left(\varphi\left(x\right)\right) + \left(1-\lambda\right)V\left(\varphi\left(x'\right)\right)\text{ ,}
        \end{split}\]
        so \(V \circ \varphi\) is strictly convex.
        
        \bigskip

        \noindent \(\left(\Rightarrow\right)\) If \(\varphi(x) = A x + b\) is such that \(A\) is not invertible, then there exist distinct \(x, x' \in \Delta\), such that \(A x = A x' = A(\lambda x + (1-\lambda) x')\) for all \(\lambda \in \left[0,1\right]\). But then for any convex \(V\) and \(\lambda \in \left[0,1\right]\), 
        \[\lambda V(\varphi(x)) + (1-\lambda) V(\varphi(x')) = V(\varphi(\lambda x + (1-\lambda) x'))\text{,}\]
        so \(U\) is not strictly non-paternalistically Blackwell monotone. \end{proof}

    \subsection{Remark \ref{remarkremark} Proof}\label{remarkremarkproof}
    We wish to show that the following updating rule is strictly Blackwell monotone.  For any experiment \(\left(\pi,S\right)\), with corresponding \(\rho\), let
\[\sigma(\pi) \coloneqq \mathbb{E}_\rho\left[\frac{\left(\mu-\mu_0\right)^{2}}{\mu_0\left(1-\mu_0\right)}\right]\text{,}\]
\[\ubar{\mu}_\pi \coloneqq \mu_0 (1- \sigma(\pi)), \quad \text{and} \quad \bar{\mu}_\pi \coloneqq \mu_0 + \sigma(\pi) (1-\mu_0)\text{,}\]
and define the experiment-dependent distortion corresponding to the updating rule by
\[\varphi_\pi(\mu) = \begin{cases}
    \ubar{\mu}_\pi, \quad &\text{if} \quad \mu \leq \ubar{\mu}_\pi\\
    \mu, \quad &\text{if} \quad \ubar{\mu}_\pi < \mu < \bar{\mu}_\pi\\
    \bar{\mu}_\pi, \quad &\text{if} \quad \mu \geq \bar{\mu}_\pi\text{.}
\end{cases}\]
    \begin{proof}
Take two experiments \(\pi\) and \(\pi'\) with \(\pi \succ \pi'\). For any decision problem and consistent decision rule, and for any \(\nu \in \left\{\ubar{\mu}_\pi,\bar{\mu}_\pi, \ubar{\mu}_{\pi'},\bar{\mu}_{\pi'}\right\}\), let \[f_\nu(\mu) = \mathbb{E}_\mu u(a^*(\nu),\theta) = \alpha_\nu \mu + \beta_\nu\text{,}\]
    for \(\alpha_\nu, \beta_\nu \in \mathbb{R}\). For each \(\tilde{\pi} \in \left\{\pi,\pi'\right\}\), define \(\bar{A}_{\tilde{\pi}} \subseteq A\) as
\[\bar{A}_{\tilde{\pi}} \coloneqq \left\{a \in A \colon \mathbb{E}_\mu u(a,\theta) = V(\mu) \text{ for some } \mu \in \left(\ubar{\mu}_{\tilde{\pi}},\bar{\mu}_{\tilde{\pi}}\right)\right\}\text{.}\]
We define on \(\left[0,1\right]\) (as \(W_{\pi}\), \(W_{\pi'}\) are defined only on \(\supp \rho\) and \(\supp \rho'\))
    \[\bar{W}_{\pi}(\mu) \coloneqq \max\left\{f_{\ubar{\mu}_\pi}(\mu), f_{\bar{\mu}_\pi}(\mu), \sup_{a \in \bar{A}_{\pi}} \mathbb{E}_\mu u(a,\theta)\right\}\text{,}\]
    and
    \[\bar{W}_{\pi'}(\mu) \coloneqq \max\left\{f_{\ubar{\mu}_{\pi'}}(\mu), f_{\bar{\mu}_{\pi'}}(\mu), \sup_{a \in \bar{A}_{\pi'}} \mathbb{E}_\mu u(a,\theta)\right\}\text{.}\]
    \(\bar{W}_{\pi}\), \(\bar{W}_{\pi'}\), and \(\bar{W}_\pi - \bar{W}_{\pi'}\) are convex, so by Theorem 3.1 in \cite{flexibilitypaper}, 
    \[\label{a1}\tag{A.1}\mathbb{E}_{\rho} W_{\pi}(\mu) = \mathbb{E}_{\rho} \bar{W}_{\pi}(\mu) \geq \mathbb{E}_{\rho'} \bar{W}_{\pi}(\mu) \geq \mathbb{E}_{\rho'} \bar{W}_{\pi'}(\mu) = \mathbb{E}_{\rho'} W_{\pi'}(\mu)\text{.}\]
    
    Strictness is easy: the standard quadratic-loss utility and unit-interval action set decision problem yields the strictly convex \(V(\mu) = - \mu (1-\mu)\) on \(\left[0,1\right]\). Let \(\pi \succ \pi'\). If \(\supp \rho' \subseteq \left[\ubar{\mu}_\pi,\bar{\mu}_\pi\right]\), then in Expression \ref{a1}, \(\mathbb{E}_{\rho} \bar{W}_{\pi}(\mu) > \mathbb{E}_{\rho'} \bar{W}_{\pi}(\mu)\). If \(\supp \rho' \not\subseteq \left[\ubar{\mu}_\pi,\bar{\mu}_\pi\right]\), then in Expression \ref{a1}, \(\mathbb{E}_{\rho'} \bar{W}_{\pi}(\mu) > \mathbb{E}_{\rho'} \bar{W}_{\pi'}\).\end{proof}

\subsection{Proposition \ref{convexprop} Proof}\label{convexpropproof}
\convexprop*

First, we note the following fact, which follows from discussion in \cite{braghieri2023biased} and is an easy consequence of the strict separating hyperplane theorem.
\begin{remark}\label{remarkunder2}
    Let binary experiment \(\left(\pi,S\right)\) induce the Bayesian distribution over posteriors \(\rho\), supported on \(\left\{\mu_1,\mu_2\right\}\) (\(\mu_1 \neq \mu_2\)). If \(U\) is grounded, \(U\) is Blackwell monotone only if \(\hat{\mu}_{2}^{\pi} \in \conv\left\{\mu_0,\mu_2\right\}\) and \(\hat{\mu}_{1}^{\pi} \in \conv\left\{\mu_1,\mu_0\right\}\).\footnote{We adapt the earlier notation so that \(\hat{\mu}^{\pi} \coloneqq \varphi_{\pi}(\mu)\).}
\end{remark}

\begin{lemma}\label{binaryground}
    If \(U\) is grounded, convex, and strictly Blackwell monotone, \(U\) is Bayes' law for any binary experiment.
\end{lemma}
\begin{proof}
    Let \(U\) be grounded. Suppose for the sake of contraposition that \(U\) is not Bayes' law for some binary experiment \(\left(\pi,S\right)\). Appealing to Remark \ref{remarkunder2}, we may specify without loss of generality that the state is binary and that the Bayesian distribution \(\rho\) has support \(\left\{0,1\right\}\) with \(\varphi_\pi(0) \equiv \gamma \in \left(0,\mu_0\right]\) and \(\varphi_\pi(1) \equiv \delta \in \left[\mu_0,1\right]\).
    
    Take five (\textbf{five!})\footnote{This quantity seems excessive, but it seems like we need all five.} additional experiments \(\left(\pi_1,S_1\right)\), \(\left(\pi_2,S_2\right)\), \(\left(\pi_3,S_3\right)\), \(\left(\pi_4,S_4\right)\), and \(\left(\pi_5,S_5\right)\), defined as follows. Let \(\rho^{i}\) (\(i \in \left\{1,\dots,5\right\}\)) denote the distribution over Bayesian posteriors induced by each \(\pi_i\). Then, for \(0 < 4 \eta < \gamma\),
    \[\begin{split}
        \supp \rho &= \left\{0,1\right\}, \ \supp \rho^1 = \left\{0, 4\eta,1\right\}, \ \supp \rho^2 = \left\{2\eta, 1\right\}, \ \supp \rho^3 = \left\{2\eta, 4\eta, 1\right\},\\
        \supp \rho^4 &= \left\{3\eta, 1\right\}, \text{ and } \supp \rho^5 = \left\{4\eta, 1\right\}\text{,}
    \end{split}\]
    with
     \[p \coloneqq \mathbb{P}_{\rho^2}(2\eta), \ \mathbb{P}_{\rho^1}(4\eta) = \mathbb{P}_{\rho^1}(0) = \frac{p}{2}, q \coloneqq \mathbb{P}_{\rho^4}(3\eta), \text{ and } \mathbb{P}_{\rho^3}(4\eta) = \mathbb{P}_{\rho^3}(2\eta) = \frac{q}{2}\text{.}\]
     Note that \[p = \frac{1-\mu_0}{1- 2\eta}\] by Bayes-plausibility. By construction, recalling that \(\succ\) denotes the ``more informative than'' relation, \[\pi \succ \pi_1 \succ \pi_2 \succ \pi_3 \succ \pi_4 \succ \pi_5\text{.}\]
    \begin{claim}\label{claim45}
    If \(U\) is convex and Blackwell monotone, \[\delta \geq \varphi_{\pi_1}(1) = \varphi_{\pi_2}(1) = \varphi_{\pi_3}(1) = \varphi_{\pi_4}(1) = \varphi_{\pi_5}(1) \geq \mu_0\text{.}\]
\end{claim}
\begin{proof}
    Throughout the proof of this claim, we can ignore the other support points (the ones less than \(\mu_0\)) of each \(\rho^i\), as Remark \ref{remarkunder2} tells us that each is mapped by the respective distortions to beliefs less than \(\mu_0\) if \(U\) is Blackwell monotone and grounded (which we specified at the start of the proof of the lemma). Moreover, by Remark \ref{remarkunder2}, \(\varphi_{\pi_i}(1) \geq \mu_0\) for each \(i\) is necessary for Blackwell monotonicity, so we also assume this.
    
    Next, we argue that \[\delta \geq \varphi_{\pi_1}(1) \geq \varphi_{\pi_2}(1) \geq \varphi_{\pi_3}(1) \geq \varphi_{\pi_4}(1) \geq \varphi_{\pi_5}(1)\] if \(U\) is Blackwell monotone. Suppose for the sake of contraposition that \(\varphi_{\pi_4}(1) < \varphi_{\pi_5}(1)\). Take a binary decision problem in which action \(a_2\) yields a state-independent payoff of \(0\) and action \(a_1\) yields a payoff, in belief \(\mu\), of\footnote{We use binary decision problems like this one throughout the proof of this lemma. The key feature of each is that the DM is indifferent between the two actions at some point that lies strictly between the two beliefs we wish to compare--in this instance, \(\varphi_{\pi_4}(1)\) and \(\varphi_{\pi_5}(1)\). Consequently, the DM takes different actions at the two beliefs: as this line is strictly less than \(0\) at \(\varphi_{\pi_4}(1)\), she takes action \(a_2\); and as the line is strictly greater than \(0\) at \(\varphi_{\pi_5}(1)\), she chooses \(a_1\).} 
    \[\mu - \frac{\varphi_{\pi_4}(1) + \varphi_{\pi_5}(1)}{2}\text{.}\]
    The DM's \textit{ex ante} expected payoff from \(\pi_4\) is \(0\) and is strictly positive from \(\pi_5\), so \(U\) is not Blackwell monotone. The analogous arguments produce the rest of the chain, pair by pair.

    Finally, we argue that \(\varphi_{\pi_1}(1) \leq \varphi_{\pi_5}(1)\) if \(U\) is convex and Blackwell monotone. Suppose for the sake of contraposition not, i.e., that  \(\varphi_{\pi_1}(1) > \varphi_{\pi_5}(1)\). By construction, \[\rho^1 =  \frac{1}{2(1-2\eta)} \rho + \frac{1-4 \eta}{2(1-2\eta)} \rho^5\text{.}\]
     Take a binary decision problem in which action \(a_2\) yields a state-independent payoff of \(0\) and action \(a_1\) yields a payoff, in belief \(\mu\), of 
    \[\mu - \frac{\varphi_{\pi_1}(1) + \varphi_{\pi_5}(1)}{2}\text{.}\]
    
    The DM's \textit{ex ante} expected payoff under \(\pi_5\) is \(v_5 \coloneqq 0\). The DM's \textit{ex ante} expected payoff under \(\pi\) is
    \[v \coloneqq \mu_0 \left(1 - \frac{\varphi_{\pi_1}(1) + \varphi_{\pi_5}(1)}{2}\right)\text{,}\]
    and it is 
    \[v_1 \coloneqq \frac{\mu_0-2\eta}{1-2\eta}\left(1 - \frac{\varphi_{\pi_1}(1) + \varphi_{\pi_5}(1)}{2}\right)\] under \(\pi_1\).

    Then,
    \[v_1 - \frac{1}{2(1-2\eta)} v - \frac{1-4 \eta}{2(1-2\eta)} v_5 = \left(\frac{\mu_0-2\eta}{1-2\eta} - \frac{\mu_0}{2(1-2\eta)}\right)\left(1 - \frac{\varphi_{\pi_2}(1) + \varphi_{\pi_3}(1)}{2}\right) > 0\text{,}\] so \(U\) is not convex.
\end{proof}
\begin{claim}\label{claima3}
If \(U\) is Blackwell monotone, for all \(\xi \in \left\{0,4\eta\right\}\) and \(\nu \in \left\{2\eta, 4 \eta\right\}\),
\[\gamma \leq \varphi_{\pi_1}(\xi) \leq \varphi_{\pi_2}(2\eta) \leq \varphi_{\pi_3}(\nu) \leq \varphi_{\pi_4}(3\eta) \leq \varphi_{\pi_5}(4\eta) \leq \mu_0 \text{.}\]
\end{claim}
\begin{proof}
    First, we argue that \[\gamma \leq \varphi_{\pi_2}(2\eta) \leq \varphi_{\pi_4}(3\eta) \leq \varphi_{\pi_5}(4\eta) \leq \mu_0\] is necessary for Blackwell monotonicity. By Remark \ref{remarkunder2}, \(\varphi_{\pi_2}(2\eta) \in \left[2\eta,\mu_0\right]\), \(\varphi_{\pi_4}(3\eta) \in \left[3\eta,\mu_0\right]\) and \(\varphi_{\pi_5}(4\eta) \in \left[4\eta,\mu_0\right]\) are necessary. Now suppose for the sake of contraposition that \(\varphi_{\pi_2}(2\eta) \in \left[2\eta,\gamma\right)\).

    Take a binary decision problem in which action \(a_2\) yields a state-independent payoff of \(0\) and action \(a_1\) a payoff (in belief \(\mu\)) of \[-\left(\mu - \frac{\varphi_{\pi_2}(2\eta) + \gamma}{2}\right)\text{.}\] Under \(\pi\), the DM gets \(0\), as she'll always take action \(a_2\); whereas, under \(\pi_2\), the DM gets a strictly positive expected payoff, a violation of Blackwell monotonicity. The analogous arguments pair-by-pair yields the necessity of the rest of the chain of inequalities.

    Second, we argue that if \(U\) is Blackwell monotone, \(\max\left\{\varphi_{\pi_1}(0), \varphi_{\pi_1}(4\eta)\right\} \leq \varphi_{\pi_2}(2\eta)\) and \(\max\left\{\varphi_{\pi_3}(2\eta), \varphi_{\pi_3}(4\eta)\right\} \leq \varphi_{\pi_4}(3\eta)\). Let us tackle the first of these. Suppose for the sake of contraposition that \(\varphi_{\pi_1}(4\eta) > \varphi_{\pi_2}(2\eta)\). Take a binary decision problem in which action \(a_2\) yields a state-independent payoff of \(0\) and action \(a_1\) a payoff (in belief \(\mu\)) of
    \[-\left(\mu - \frac{\varphi_{\pi_2}(2\eta) + \varphi_{\pi_1}(4\eta)}{2}\right)\text{.}\]
    The DM's expected payoff under \(\pi_1\) is bounded above by
    \[v_1 \coloneqq \frac{p}{2}\left(\frac{\varphi_{\pi_2}(2\eta) + \varphi_{\pi_1}(4\eta)}{2}\right)\text{;}\]
    whereas her expected payoff under \(\pi_2\) is 
    \[v_2 \coloneqq -p\left(2\eta - \frac{\varphi_{\pi_2}(2\eta) + \varphi_{\pi_1}(4\eta)}{2}\right)\text{.}\]
    Then, \(v_2 - v_1\) equals
    \[\frac{p}{2} \left(\frac{\varphi_{\pi_2}(2\eta) + \varphi_{\pi_1}(4\eta)}{2} - 4\eta\right) > \frac{p}{2}\left(\varphi_{\pi_2}(2\eta) - 4\eta\right) > 0\text{,}\]
    a violation of Blackwell monotonicity. The analogous argument yields the necessity of \(\varphi_{\pi_1}(0) \leq \varphi_{\pi_2}(2\eta)\). Likewise, the proof can repeated to argue \(\max\left\{\varphi_{\pi_3}(2\eta), \varphi_{\pi_3}(4\eta)\right\} \leq \varphi_{\pi_4}(3\eta)\).

    Third, we argue that if \(U\) is Blackwell monotone, \(\gamma \leq \min\left\{\varphi_{\pi_1}(0), \varphi_{\pi_1}(4\eta)\right\}\) and \(\varphi_{\pi_2}(2\eta) \leq \min\left\{\varphi_{\pi_3}(2\eta), \varphi_{\pi_3}(4\eta)\right\}\). Let us tackle the first of these. Suppose for the sake of contraposition that \(\varphi_{\pi_1}(0) < \gamma\). Take a binary decision problem in which action \(a_2\) yields a state-independent payoff of \(0\) and action \(a_1\) a payoff (in belief \(\mu\)) of
    \[-\left(\mu - \frac{\max\left\{\varphi_{\pi_1}(0),4\eta\right\} + \gamma}{2}\right)\text{.}\]
    The DM's \textit{ex ante} expected payoff under \(\pi_1\) is strictly positive, whereas it is \(0\) under \(\pi\), a violation of Blackwell monotonicity. Note that we have assumed without loss of generality (or else \(U\) would not be Blackwell monotone) that \(\varphi_{\pi_1}(4 \eta) \in \left[0,\mu_0\right]\).

    Now suppose for the sake of contraposition that \(\varphi_{\pi_1}(4 \eta) < \gamma \leq \varphi_{\pi_1}(0)\). If \(\varphi_{\pi_1}(4 \eta) < 4 \eta\), \(U\) is not Blackwell monotone so we specify \(\varphi_{\pi_1}(4 \eta) \geq 4 \eta\). Take a binary decision problem in which action \(a_2\) yields a state-independent payoff of \(0\) and action \(a_1\) a payoff (in belief \(\mu\)) of
    \[-\left(\mu - \frac{\varphi_{\pi_1}(4\eta) + \gamma}{2}\right)\text{.}\]
    The DM's \textit{ex ante} expected payoff under \(\pi_1\) is strictly positive, whereas it is \(0\) under \(\pi\), a violation of Blackwell monotonicity.

    This proof can be repeated to argue \(\varphi_{\pi_2}(2\eta) \leq \min\left\{\varphi_{\pi_3}(2\eta), \varphi_{\pi_3}(4\eta)\right\}\).\end{proof}
    \begin{claim}\label{claim47}
If \(U\) is convex and Blackwell monotone,
\[\varphi_{\pi_2}(2\eta) = \varphi_{\pi_3}(2\eta) = \varphi_{\pi_3}(4\eta) = \varphi_{\pi_4}(3\eta) = \varphi_{\pi_5}(4\eta) \text{.}\]
\end{claim}
\begin{proof}
    From Claim \ref{claima3}, it suffices to show that \(\max\left\{\varphi_{\pi_1}(0), \varphi_{\pi_1}(4\eta)\right\} \geq \varphi_{\pi_5}(4\eta)\). Suppose for the sake of contraposition not, i.e., that \(\max\left\{\varphi_{\pi_1}(0), \varphi_{\pi_1}(4\eta)\right\} < \varphi_{\pi_5}(4\eta)\). Take a binary decision problem in which action \(a_2\) yields a state-independent payoff of \(0\) and action \(a_1\) a payoff (in belief \(\mu\)) of
    \[-\left(\mu - \frac{\max\left\{\varphi_{\pi_1}(0), \varphi_{\pi_1}(4\eta)\right\} + \varphi_{\pi_5}(4\eta)}{2}\right)\text{.}\]

    The DM's \textit{ex ante} expected payoff under \(\pi_5\) is \(v_5 \coloneqq 0\). Her \textit{ex ante} expected payoff under \(\pi\) is
    \[v \coloneqq (1-\mu_0)\frac{\max\left\{\varphi_{\pi_1}(0), \varphi_{\pi_1}(4\eta)\right\} + \varphi_{\pi_5}(4\eta)}{2}\text{,}\]
    and it is 
    \[\begin{split}
        v_1 &\coloneqq -\frac{1-\mu_0}{1-2\eta}\left(2\eta - \frac{\max\left\{\varphi_{\pi_1}(0), \varphi_{\pi_1}(4\eta)\right\} + \varphi_{\pi_5}(4\eta)}{2}\right) = \frac{1}{1-2\eta} v - 2 \eta \frac{1-\mu_0}{1-2\eta}\text{,}
    \end{split}\]
    under \(\pi_1\). Then,
    \[v_1 - \frac{1}{2(1-2\eta)} v - \frac{1-4 \eta}{2(1-2\eta)} v_5 = \frac{1}{2(1-2\eta)}\left(v - 4 \eta (1-\mu_0)\right) > 0\text{,}\] so \(U\) is not convex.\end{proof}

    \begin{claim}\label{claima5}
    If \(U\) is convex, \(U\) is not strictly Blackwell monotone.
\end{claim}
\begin{proof}
    By Claims \ref{claim45} and \ref{claim47}, if \(U\) is convex and Blackwell monotone,  \(\varphi_{\pi_4}(1) = \varphi_{\pi_3}(1)\) and \(\varphi_{\pi_3}(2\eta) = \varphi_{\pi_3}(4\eta) = \varphi_{\pi_4}(3\eta)\). By construction \(\pi_3 \succ \pi_4\), yet for any decision problem and consistent choice of action, \(\pi_3\) and \(\pi_4\) yield the DM the same \textit{ex ante} expected payoff, so \(U\) is not strictly Blackwell monotone.\end{proof}
    This concludes the proof of the lemma.\end{proof}

    \begin{lemma}\label{fullground}
    If \(U\) is grounded, convex, and strictly Blackwell monotone, \(U\) is Bayes' law for any experiment.\end{lemma}
    \begin{proof}
    Let \(U\) be grounded and convex. Take an arbitrary non-binary experiment \(\left(\pi,S\right)\) and let \(\rho\) be the induced Bayesian distribution over posteriors. Suppose for the sake of contraposition that there exists some \(x \in \supp \rho\) for which \(\varphi_\pi(x) \equiv \hat{x} \neq x\). Consequently, there exists \(\alpha \in \mathbb{R}^{n-1}\) and \(\beta \in \mathbb{R}\) such that 
    \[\alpha x + \beta > 0 > \alpha \hat{x} + \beta\text{.}\]

    Without loss of generality we assume \(\varphi_\pi(\mu) = \mu\) for all \(\mu \in \supp \rho \setminus \left\{x\right\}\). We define \[M \coloneqq \left\{\mu \in \supp \rho \colon \alpha \mu + \beta \geq 0\right\}\text{.}\]
    If \(\supp \rho \setminus M = \emptyset\), we are done, as Remark \ref{remarkunder2} tells us that \(U\) is not Blackwell monotone. 
    
    Suppose that \(\supp \rho \setminus M \neq \emptyset\), and take the binary Bayesian distribution over posteriors obtained by collapsing the posteriors in \(M\) and the posteriors in \(\supp \rho \setminus M\) to their respective barycenters (under \(\rho\)). Call this Bayesian distribution over posteriors \(\rho'\) and denote the experiment that induces it \(\pi'\). Naturally \(\pi \succ \pi'\). Moreover, if \(\varphi_{\pi'}(\mu) \neq \mu\) for any of the two posteriors in \(\supp \rho'\), Lemma \ref{binaryground} tells us that \(U\) is not strictly Blackwell monotone. 

    Suppose, therefore, that \(\varphi_{\pi'}(\mu) = \mu\) for all \(\mu \in \supp \rho'\). But then the DM must strictly prefer \(\pi'\) to \(\pi\), as she makes an error with strictly positive probability under \(\pi\). We conclude that \(U\) is not Blackwell monotone.\end{proof}

    \subsection{Proposition \ref{focusprop} Proof}\label{focuspropproof}
    \focusprop*
    The proposition is the product of Lemmas \ref{lemma6} and \ref{lemma7}. As ever, we need only to establish the necessity of Bayes' law.
    \begin{lemma}\label{lemma6}
        If \(U\) is grounded, focused, and strictly Blackwell monotone, \(U\) is Bayes' law for any binary experiment.
    \end{lemma}
    \begin{proof}
        We begin just as in the proof of Lemma \ref{binaryground}. Let \(U\) be grounded and suppose for the sake of contraposition that \(U\) is not Bayes' law for some binary experiment \(\left(\pi,S\right)\). As before, Remark \ref{remarkunder2} enables us to specify without loss of generality that the state is binary and that the Bayesian distribution \(\rho\) has support \(\left\{0,1\right\}\) with \(\varphi_\pi(0) \equiv \gamma \in \left(0,\mu_0\right]\) and \(\varphi_\pi(1) \equiv \delta \in \left[\mu_0,1\right]\). 

        We take the five additional experiments constructed in Lemma \ref{binaryground}'s proof and add a sixth to the mix; \(\left(\pi_6,S_6\right)\), where \(\pi_6\) is binary and obtained by combining the columns of \(\pi_1\) that produce Bayesian posteriors \(0\) and \(1\) (adding the two elements in each row of the two columns together). Consequently, the Bayesian distribution over posteriors corresponding to \(\pi_6\) has support \(\left\{4\eta, \tau\right\}\), where \(\mathbb{P}_{\rho^1}(4\eta) = \mathbb{P}_{\rho^6}(4\eta)\). Moreover, if \(U\) is focused, \(\varphi_{\pi_6}(4\eta) = \varphi_{\pi_1}(4\eta)\). Note also that 
        \[\pi \succ \pi_1 \succ \pi_2 \succ \pi_3 \succ \pi_4 \succ \pi_5 \succ \pi_6\text{.}\]

        Next, 
        \begin{claim}\label{claima6}
        If \(U\) is Blackwell monotone, for all \(\xi \in \left\{0,4\eta\right\}\) and \(\nu \in \left\{2\eta, 4 \eta\right\}\),
        \[\gamma \leq \varphi_{\pi_1}(\xi) \leq \varphi_{\pi_2}(2\eta) \leq \varphi_{\pi_3}(\nu) \leq \varphi_{\pi_4}(3\eta) \leq \varphi_{\pi_5}(4\eta) \leq \varphi_{\pi_6}(4\eta) \leq \mu_0 \text{.}\]
\end{claim}
\begin{proof}
    Claim \ref{claima3} states the entire chain other than \(\varphi_{\pi_5}(4\eta) \leq \varphi_{\pi_6}(4\eta)\), but the logic for that is identical, so we can conclude it as well.\end{proof}
\begin{claim}\label{claima7}
    If \(U\) is focused and Blackwell monotone,
    \[\varphi_{\pi_2}(2\eta) = \varphi_{\pi_3}(2\eta) = \varphi_{\pi_3}(4\eta) = \varphi_{\pi_4}(3\eta) = \varphi_{\pi_5}(4\eta) = \varphi_{\pi_6}(4\eta) \text{.}\]
\end{claim}
\begin{proof}
    From Claim \ref{claima6}, it suffices to show that \(\varphi_{\pi_1}(4\eta) = \varphi_{\pi_6}(4\eta)\). We have already observed that if this is not true then \(U\) is not focused, so we are done.\end{proof}
\begin{claim}\label{claima8}
    If \(U\) is focused, \(\varphi_{\pi_3}(1) = \varphi_{\pi_4}(1)\).
\end{claim}
\begin{proof}
    This follows immediately from the fact that \(\pi_4\) is obtained from \(\pi_3\) by combining two columns, leaving one (the one that yields posterior \(1\)) untouched.\end{proof}
\begin{claim}\label{claima9}
    If \(U\) is focused, \(U\) is not strictly Blackwell monotone.
\end{claim}
\begin{proof}
    This proof mimics that of Claim \ref{claima5}. By Claims \ref{claima7} and \ref{claima8}, if \(U\) is focused and Blackwell monotone, \(\varphi_{\pi_4}(1) = \varphi_{\pi_3}(1)\) and \(\varphi_{\pi_3}(2\eta) = \varphi_{\pi_3}(4\eta) = \varphi_{\pi_4}(3\eta)\). \(\pi_3 \succ \pi_4\), yet for any decision problem and consistent choice of action, \(\pi_3\) and \(\pi_4\) produce the same \textit{ex ante} expected payoff, so \(U\) is not strictly Blackwell monotone.\end{proof}
    We conclude the lemma.\end{proof}
    \begin{lemma}\label{lemma7}
    If \(U\) is grounded, focused, and strictly Blackwell monotone, \(U\) is Bayes' law for any experiment.
\end{lemma}
\begin{proof}
    The proof is identical to Lemma \ref{fullground}'s proof and so we leave it out.
\end{proof}

\section{Forecasting Errors}\label{forecast}

        An alternative way of evaluating the value of information is one in which we allow for forecast errors from the \textit{ex ante} perspective. In this case, we specify that the anticipated distribution over posteriors is \(\hat{\rho}\) and so the \textit{ex ante} expected payoff is \(\mathbb{E}_{\hat{\rho}}V(\hat{\mu})\). We also assume that the DM understands the martingality of beliefs, so that the anticipated distribution over posteriors, \(\hat{\rho}\), has mean \(\mu_0\).

        Consequently, an updating rule is Blackwell monotone (in this sense) if \(\pi \succeq \pi'\) implies \(\mathbb{E}_{\hat{\rho}}\mathbb{E}_{\hat{\mu}}u(a,\theta) \geq \mathbb{E}_{\hat{\rho}'}\mathbb{E}_{\hat{\mu}}u(a,\theta)\) for any decision problem. Then, we have the following easy result:
        \begin{remark}
    An updating rule is Blackwell monotone if and only if \(\pi \succeq \pi'\) implies \(\hat{\rho}\) is a mean-preserving spread of \(\hat{\rho}'\). \end{remark}

        \section{Two States and Extreme-Belief Aversion}\label{appx}
    Let there be two states (\(n = 2\)). Recall that a DM with updating rule \(U\) displays extreme-belief aversion if there exist two intervals \(C_1 \coloneqq (0,c)\) and \(C_2 \coloneqq (d,1)\) (with \(0 \leq c \leq d \leq 1\)) such that \(\varphi(\mu) = c\) for all \(\mu \in C_1\),  \(\varphi(\mu) = d\) for all \(\mu \in C_2\)  and \(\varphi(\mu) = \mu\) for all \(\mu \in \left[c,d\right]\).

Proposition \ref{propxba} tells us that extreme-belief aversion is implied by Blackwell monotonicity. Now let us show that it is sufficient when \(\varphi\) is continuous.
\begin{proposition}
    Let \(n = 2\). An updating rule \(U\) that systematically distorts beliefs according to a continuous distortion \(\varphi\) is Blackwell monotone if and only if the DM displays extreme-belief aversion.
\end{proposition}
\begin{proof}
Suppose \(U\) displays extreme-belief aversion. Because \(\varphi\) is continuous, \(\varphi(0) = c\) and \(\varphi(d) = 1\). Define \(\bar{A} \subseteq A\) as
\[\bar{A} \coloneqq \left\{a \in A \colon \mathbb{E}_\mu u(a,\theta) = V(\mu) \text{ for some } \mu \in \left(c,d\right)\right\}\text{.}\]
Given the DM's consistent decision rule, let \(f_c(\mu) \coloneqq \alpha_c \mu + \beta_c\) and \(f_d (\mu) \coloneqq \alpha_d \mu + \beta_d\) be the expected payoffs (in belief \(\mu\)) to the actions the DM takes at beliefs \(c\) and \(d\), respectively, where \(\alpha_i, \beta_i \in \mathbb{R}\) for \(i \in \left\{c,d\right\}\). By interim optimality and the definition of \(\bar{A}\), \(f_c(\mu) \geq \sup_{a \in \bar{A}}\mathbb{E}_\mu u(a,\theta)\) for all \(\mu \leq c\). Likewise, \(f_d(\mu) \geq \sup_{a \in \bar{A}}\mathbb{E}_\mu u(a,\theta)\) for all \(\mu \geq d\). Thus, \(W(\mu) = \max\left\{f_c(\mu), f_d(\mu), \sup_{a \in \bar{A}} \mathbb{E}_\mu u(a,\theta)\right\}\), which is convex.\end{proof}
Here is an example of such a rule.
\begin{example} There are two states, \(\Theta = \left\{0,1\right\}\), and the set of actions is the unit interval, \(A = \left[0,1\right]\). The DM's utility function is the standard ``quadratic loss'' utility, translated up by \(.3\) (to make the graph prettier): \(u\left(a,\theta\right) = - \left(a-\theta\right)^2 + .3\). Accordingly, \(V\left(x\right) = - x \left(1-x\right) + .3\). Here is an \href{https://www.desmos.com/calculator/lpmuznley9}{Interactive Link}, where one can adjust the parameters, \(c\) and \(d\), that specify the continuous extreme-belief averse rule by moving the corresponding sliders.\end{example}
The continuity conditional in the proposition merely absolves us from needing to specify where the vertices of the simplex are mapped by \(\varphi\). Without imposing continuity, in the two state setting, Blackwell monotonicity is equivalent to the DM displaying extreme-belief aversion with \(\varphi(0) \leq c\) and \(\varphi(1) \geq d\). A previous draft of this paper (v4 on ArXiv) contains this version of the result in Theorem 3.1.

\bibliography{sample.bib}

\newpage

\section{Online Appendix}

In the main text, we have fixed the DM's full-support prior, but not the decision problem--so that a Blackwell-monotone updating rule is such that more information is more valuable in every decision problem. What if we instead start with a fixed decision problem? Our search, therefore, is for updating rules that generate a positive marginal value of information \textit{for the particular decision problem under consideration}.

Fix a full-support prior \(\mu_0 \in \Delta^{\circ}\) and a finite-action decision problem with at least two actions in which no action is weakly dominated. Let the set of actions be \(A = \left\{a_1, \dots, a_m\right\}\) (with \(m \in \mathbb{N}\), \(m \geq 2\)). Recalling the definition of value function \(V \coloneqq \max_{a \in A}\mathbb{E}_\mu u(a,\theta)\), a \textcolor{OrangeRed}{Problem-Specific Updating Rule} (that systematically distorts beliefs), \(U_V\), is an updating rule for which the DM's posterior is a posterior-separable function of the Bayesian posterior \(\varphi_V \colon \Delta \to \Delta\).

We maintain the specification that the DM's choice at every belief \(\mu \in \Delta\) is consistent (and interim-optimal); \textit{viz.,} depends only on the realized posterior.
\begin{definition}
    A problem-specific updating rule, \(U_V\), is \textcolor{OrangeRed}{Blackwell Monotone} if for any consistent choice of action \(a^* \colon \Delta \to A\), a DM's \textit{ex ante} expected utility from observing experiment \((\pi, S)\) is higher than from observing \((\pi', S')\) if \(\pi \succeq \pi'\), where \(\succeq\) is the (Blackwell) order over experiments.
\end{definition}

A \textcolor{OrangeRed}{Menu} of actions is some nonempty subset \(\bar{A} \subseteq A\). Given a menu of actions \(\bar{A}\), let \(V_{\bar{A}}(\mu) \coloneqq \max_{a \in \bar{A}}\mathbb{E}_\mu u(a,\theta)\) be the DM's value function in Bayesian-belief \(\mu\). We say a problem-specific updating rule generates an effective menu on \(Y \subseteq \Delta\) if for any consistent choice, there exists some \(\bar{A} \subseteq A\) such that \(W = V_{\bar{A}}\) on \(Y\).

\begin{restatable}{theorem}{updateproblem}
    \label{updateproblem}
    A problem-specific updating rule that systematically distorts beliefs is Blackwell monotone if it generates an effective menu on \(\Delta\). It is Blackwell monotone only if it generates an effective menu on \(\Delta^{\circ}\).
\end{restatable}
\begin{proof}
    \(\left(\Leftarrow\right)\) The convexity of \(V_{\bar{A}}\) gives us Blackwell monotonicity.

    \bigskip

    \noindent \(\left(\Rightarrow\right)\) If \(U_V\) is Blackwell monotone, \(W\) is convex on \(\Delta\). This implies that on \(\Delta^{\circ}\), \[W(\mu) = \sup_{j \in J}\left\{\alpha_j \mu + \beta_j\right\}\text{,}\] for some index set \(J\), where \(\alpha_j \in \mathbb{R}^{n-1}\) and \(\beta_j \in \mathbb{R}\) for all \(j \in J\). For each \(a \in A\), let \(f_a (\mu) \coloneqq \alpha_a \mu + \beta_a\) denote its payoff in belief \(\mu\) (where \(\alpha_a \in \mathbb{R}^{n-1}\) and \(\beta_a \in \mathbb{R}\) for all \(a \in A\)). As the set of beliefs at which \(f_a (\mu) \neq f_{a'}(\mu)\) for any distinct pair \(a, a' \in A\) is dense in \(\Delta^{\circ}\), 
    \[W(\mu) = \sup_{j \in J}\left\{\alpha_j \mu + \beta_j\right\} = \max_{a \in \bar{A}} f_{a}(\mu) = V_{\bar{A}}(\mu)\text{,}\] for some \(\bar{A} \subseteq A\) on a dense subset of \(\Delta^{\circ}\), implying that \(W = V_{\bar{A}}\) on \(\Delta^{\circ}\).\end{proof}
The essence of this theorem is that for a fixed decision problem, updating rules that respect the Blackwell order must be behaviorally equivalent to a DM updating to full-support posteriors correctly but ignoring particular actions except possibly on a measure-zero subset of beliefs. Naturally, many different distortions can produce problem-specific Blackwell monotonicity. However, the behavior induced by such rules must take a particular form; namely, the almost-total ignorance of some subset of actions.

\end{document}